\documentclass[reqno]{amsart}
\usepackage{amsmath,amsthm,amssymb,mathrsfs,stmaryrd,eucal,mathbbol} 
\usepackage[all,cmtip]{xy}
\usepackage{hyperref}
\usepackage{tikz}

\numberwithin{equation}{section}

\usepackage{pgfplots}
\usepgfplotslibrary{dateplot}

\newcommand\void[1]{}

\newcommand{\C}{\mathbb{C}}

\newcommand{\Z}{\mathbb{Z}}

\newcommand{\CA}{\mathcal{A}}
\newcommand{\CB}{\mathcal{B}}
\newcommand{\CC}{\mathcal{C}}
\newcommand{\CD}{\mathcal{D}}
\newcommand{\CE}{\mathcal{E}}

\newcommand{\CL}{\mathcal{L}}
\newcommand{\CM}{\mathcal{M}}
\newcommand{\CN}{\mathcal{N}}

\newcommand{\CP}{\mathcal{P}}

\newcommand{\CR}{\mathcal{R}}

\newcommand{\CU}{\mathcal{U}}
\newcommand{\CV}{\mathcal{V}}
\newcommand{\CW}{\mathcal{W}}
\newcommand{\CX}{\mathcal{X}}
\newcommand{\CY}{\mathcal{Y}}

\newcommand{\FZ}{\mathfrak{Z}}

\newcommand{\Cat}{\mathrm{Cat}}
\newcommand{\Hilb}{\mathrm{Hilb}}

\newcommand{\Vect}{\mathrm{Vec}}
\newcommand{\op}{\mathrm{op}}
\newcommand{\rev}{\mathrm{rev}}

\newcommand{\TO}{\mathcal{QL}}
\newcommand{\TOsk}{\mathcal{QL}_{\mathrm{sk}}}
\newcommand{\CXsk}{\mathcal{X}_{\mathrm{sk}}}
\newcommand{\CXlqs}{\mathcal{X}_{\mathrm{lqs}}}

\newcommand{\KarCat}{\mathrm{KarCat}}
\newcommand{\Kar}{\mathrm{Kar}}

\newcommand{\Bord}{\mathbf{Bord}}

 \DeclareMathOperator{\Hom}{Hom}
 \DeclareMathOperator{\End}{End}

 \DeclareMathOperator{\Id}{Id}

 \DeclareMathOperator{\chara}{char}

 \DeclareMathOperator{\ev}{ev}
 
 \DeclareMathOperator{\one}{\mathbf1}

 \DeclareMathOperator{\Fun}{Fun}

 \DeclareMathOperator{\Mod}{Mod}
 \DeclareMathOperator{\LMod}{LMod}
 \DeclareMathOperator{\RMod}{RMod}
 \DeclareMathOperator{\BMod}{BMod}
 
 \DeclareMathOperator{\Mor}{Mor}

 \DeclareMathOperator{\Rep}{Rep}

\newtheorem{thm}{Theorem}[section]
\newtheorem{lem}[thm]{Lemma}
\newtheorem{prop}[thm]{Proposition}
\newtheorem{cor}[thm]{Corollary}

\newtheorem{conj}[thm]{Conjecture}

\theoremstyle{definition}
\newtheorem{defn}[thm]{Definition}

\newtheorem{exam}[thm]{Example}
\newtheorem{rem}[thm]{Remark}

\newtheorem{hyp}[thm]{Hypothesis}
\theoremstyle{remark}

\newcommand\arXiv[1]{\href{http://arxiv.org/abs/#1}{arXiv:#1}}

\newcommand\condense{\mathrel{\,\hspace{.75ex}\joinrel\rhook\joinrel\hspace{-.75ex}\joinrel\rightarrow}}

\begin{document}

\title{Categories of quantum liquids I}
\maketitle

\begin{center}
{\large
Liang Kong$^{a,b,c}$,\,
Hao Zheng$^{a,c,d,e,f}$\,
~\footnote{Emails:
{\tt  kongl@sustech.edu.cn, haozheng@mail.tsinghua.edu.cn}}}
\\[1em]
{\small $^a$ Shenzhen Institute for Quantum Science and Engineering, \\
Southern University of Science and Technology, Shenzhen 518055, China
\\[0.5em]
$^b$ International Quantum Academy,  Shenzhen 518048, China
\\[0.5em]
$^c$ Guangdong Provincial Key Laboratory of Quantum Science and Engineering, Southern University of Science and Technology, Shenzhen 518055, China
\\[0.5em]
$^d$ Institute for Applied Mathematics, Tsinghua University, Beijing 100084, China
\\[0.5em]
$^e$ Beijing Institute of Mathematical Sciences and Applications, Beijing 101408, China
\\[0.5em]
$^f$ Department of Mathematics, Peking University, Beijing 100871, China
}
\end{center}

\vspace{0.1cm}
\begin{abstract}
We develop a mathematical theory of separable higher categories based on Gaiotto and Johnson-Freyd's work on condensation completion. Based on this theory, we prove some fundamental results on $E_m$-multi-fusion higher categories and their higher centers. We also outline a theory of unitary higher categories based on a $*$-version of condensation completion. After these mathematical preparations, based on the idea of topological Wick rotation, we develop a unified mathematical theory of all quantum liquids, which include topological orders, SPT/SET orders, symmetry-breaking orders and CFT-like gapless phases. We explain that a quantum liquid consists of two parts, the topological skeleton and the local quantum symmetry, and show that all $n$D quantum liquids form a $*$-condensation complete higher category whose equivalence type can be computed explicitly from a simple coslice 1-category.
\end{abstract}

\tableofcontents

\section{Introduction} \label{sec:intro}

\subsection{Motivations}
In recent years, the study of topological orders has become one of the most active fields of research in condensed matter physics and mathematical physics (see recent reviews \cite{Wen17,Wen19} and references therein). Topological orders, a notion which was first introduced by Wen \cite{Wen90}, are certain gapped quantum phases often characterized by their internal structures and properties, such as the robust ground state degeneracy and topological defects. A modern point of view from mathematics, however, suggests that a rather complete understanding of a topological order can be achieved only by studying the category of all topological orders. Throughout this work, $n$D represents the spacetime dimension. 

\smallskip
The study of the categories of topological orders in all dimensions was initiated in \cite{KWZ15}. But a couple of fundamental ingredients were missing there. 
\begin{enumerate}
\item The first missing ingredient is a proper definition of a multi-fusion $n$-category for $n>1$. In 2018, Douglas and Reutter found the proper definition for $n=2$ cases in \cite{DR18}. In 2020, Johnson-Freyd further generalized the definition to $n>2$ cases in \cite{JF20},  a work which was based on an earlier and important work by Gaiotto and Johnson-Freyd on the so-called `condensation completion' or `Karoubi completion' in higher categories \cite{GJF19}. In this work, we introduce the notion of a unitary (multi-fusion) $n$-category based on a $*$-version of the condensation completion. 

\item The second missing ingredient is a theory of gapless defects. Since many topological orders have topologically protected gapless boundaries, the category of topological orders without gapless defects is incomplete (see \cite{KLWZZ20b} for a recent discussion). A unified mathematical theory of gapped and gapless boundaries of 3D topological orders was developed recently in \cite{KZ18b,KZ20,KZ21}, and is ready to be generalized to higher dimensions \cite{KZ20,KZ21}. It was confirmed in a concrete lattice model and was shown to be useful in the study of purely edge phase transition \cite{CJKYZ20}. This theory suggests that there is a unified mathematical theory of a large family of gapped and gapless quantum phases far beyond topological orders. We name this family of quantum phases {\it quantum liquid phases} or just {\it quantum liquids} for simplicity. 
\end{enumerate}
In this work and its sequels \cite{KZ21,KZ22}, we incorporate above two missing ingredients to develop a mathematical theory of all quantum liquids.

\smallskip
The name `quantum liquids' is motivated by the existing physical notion of a `gapped quantum liquid' \cite{ZW15,SM16} because the later notion should coincide with that of a gapped `quantum liquid' in our sense. By a physical notion, we mean it can be defined microscopically via lattice models. `Gapped quantum liquids' include topological orders, symmetry protected/enriched topological (SPT/SET) orders \cite{GW09,CGW10,CLW11,CGLW13} and symmetry-breaking orders\footnote{Symmetry-breaking orders include the usual symmetry-breaking orders in the old Landau's paradigm and those obtained from SPT/SET orders by partially breaking the symmetries (see Example \ref{exam:2D-classification}).}. Examples of gapped non-liquid phases were also known (see, for example, \cite{C05,H11}). For gapless quantum liquids, unfortunately, no microscopic definition is available\footnote{A possibly related  notion of a {\em quantum order} was introduced by Wen in \cite{Wen02}.}. However, the notion of a quantum phase is a macroscopic one. In principle, the mathematical characterization of a quantum liquid can be obtained by summarizing all its macroscopic observables, and can be viewed as the macroscopic definition of a quantum liquid. Indeed, we have developed such a macroscopic and unified mathematical theory of all gapped and gapless boundaries of 3D topological orders in \cite{KZ18b,KZ20,KZ21}. We have also proposed that it can be generalized to higher dimensions thus gives a much bigger macroscopic theory of (newly named) quantum liquids. In this work, we show that quantum liquids indeed include all `gapped quantum liquids' in the sense of \cite{ZW15,SM16} and certain liquid-like gapless phases (see Remark \ref{rem:liquid-like}). By developing this bigger theory, we automatically obtain the macroscopic definition of a quantum liquid.

One of the guiding principles of our theory is the so-called topological Wick rotation, a notion which was introduced in \cite{KZ20} and is reviewed in Subsection \ref{sec:twr}. Another guiding principle is the boundary-bulk relation (i.e. the bulk is the center of a boundary) \cite{KWZ15,KWZ17}. Therefore, we require a quantum liquid to satisfy the conditions for the proof of the boundary-bulk relation in \cite{KWZ17} to work: (1) a potentially anomalous quantum liquid has a unique 1-dimension-higher anomaly-free bulk; (2) fusions among quantum liquids and defects are well-defined; (3) dimensional reductions via fusions are independent of the orders of the fusions as the $\otimes$-excision property of factorization homology (see \cite{AF20} for a review and references therein). Hence, the boundary-bulk relation holds for all quantum liquids by definition.

We explain in Subsection \ref{sec:twr} that a quantum liquid $\CX$ contains two types of data $\CXlqs$ and $\CXsk$, i.e. $\CX=(\CXlqs,\CXsk)$, where 
\begin{enumerate}
\item $\CXlqs$ contains the information of local observables such as onsite symmetries, the OPE of local fields or the nets of local operators, and is called the {\em local quantum symmetry} of $\CX$;

\item $\CXsk$ contains the categorical information of all topological defects 
and is called the {\em topological skeleton} of $\CX$. 

\end{enumerate}
In this work, we focus our study on $\CXsk$ and the higher category $\TOsk^n$ of the topological skeletons of $n$D quantum liquids. In \cite{KZ21b}, we further extend our study on $\CXsk$ and prove the functoriality of centers and develop the theory of minimal modular extensions for higher categories. 
In \cite{KZ22}, we study $\CXlqs$ and provide a rather complete mathematical theory of quantum liquids.


\begin{rem} \label{rem:liquid-like}
The `liquid-like' property means that the phase is `soft' enough so that its effective field theory does not rigidly depend on the local geometry of spacetime. By `not rigidly depend', we mean the dependence of the metric of the spacetime is either trivial or in a finite and controllable way\footnote{It is very challenging to formulate this requirement on the level of lattice models.} as in 2D rational conformal field theories (CFT). In particular, bending the phase (or defects in it) makes no difference in the long wave length limit. This already implies `fully dualizability'. Topological orders, SPT/SET orders and symmetry-breaking orders are examples of quantum liquids because they are topological and independent of the metric. 
2D rational CFT's, as fully dualizable QFT's, are examples of gapless quantum liquids. Note that a 2D rational CFT only `softly' depends on the spacetime metric or, more precisely, it depends covariantly on the conformal structures of Riemann surfaces. This fact makes a 2D rational CFT looks rather `topological'. 
\end{rem}

\subsection{Results and layout}
In this work, based on the works \cite{GJF19,JF20}, we further develop the mathematical theory of multi-fusion $n$-categories but from a slightly different point of view \cite{KWZ15}. We first develop a mathematical theory of separable $n$-categories. Among many other results, we show that the category of all separable $n$-categories can be identified with $(n+1)\Vect$. Based on this theory, we define and study $E_m$-multi-fusion $n$-categories. We recover some fundamental results on multi-fusion $n$-categories in \cite{JF20} and prove some new results on $E_m$-multi-fusion $n$-categories. In particular, we give a structure theorem of multi-fusion $n$-categories in Proposition \ref{prop:mfc-morita}, and provide simple criterions for an $E_m$-monoidal $n$-category to be $E_m$-fusion in Proposition \ref{prop:indecom-fus}, Corollary \ref{cor:indecom-emfus} and Theorem \ref{thm:sep-mf}, and prove that the looping of an $E_m$-multi-fusion $n$-category is an $E_{m+1}$-fusion ($n-$1)-category in Corollary \ref{cor:em-omega}. We also prove some useful results on $E_n$-centers. For example, Theorem \ref{thm:mfc-center} characterizes the $E_m$-center in terms of the $E_0$-center, and Proposition \ref{prop:center-square=1} gives a few results on $E_1$-centers and $E_0$-centers, which have important applications in physics. We also compute the delooping of certain (co)slice categories in Proposition \ref{prop:slice-sigma} and Proposition \ref{prop:slice-sigma-add}. These results play an important role in our study of $n$D quantum liquids in Subsection \ref{sec:tosk}. For physical applications, we propose a $*$-version of condensation completion and give some results on unitary $n$-categories that are parallel to non-unitary ones in Section \ref{sec:unitary}. 

\smallskip
After above mathematical preparation, we study quantum liquids in Section \ref{sec:TO}. We first recall the mathematical theory of the gapped/gapless boundaries of 3D topological orders and the idea of topological Wick rotation developed in \cite{KZ18b,KZ20,KZ21}. The proposal that topological Wick rotation also works in higher dimensions \cite{KZ20,KZ21} leads to a much bigger theory of yet-unknown quantum phases, which are named `quantum liquids' (see Hypothesis\ \ref{hyp:quantum-liquid}). We explain that a quantum liquid can be described by two types of macroscopic observables: local quantum symmetries and topological skeletons. By slightly generalizing the results in \cite{KLWZZ20a}, we obtain a unified theory of SPT/SET orders and symmetry-breaking orders (see Theorem\,\ref{thm:sb-order}). Then, applying topological Wick rotation to this unified theory, we obtain the physical characterization of SPT/SET orders and symmetry-breaking orders in terms of enriched higher categories. This also means that quantum liquids include all SPT/SET orders and symmetry-breaking orders. In 2D, our theory reproduces the classification of all 2D gapped quantum phases with finite onsite symmetries obtained earlier via a microscopic approach \cite{CGW10b,SPGC11} (see Example \ref{exam:2D-classification}). It is further confirmed by a direct study of the Ising chain and the Kitaev chain in \cite{KWZ21} and more general 2D models in \cite{XZ22}. Next, we partially define the category $\TO^n$ of $n$D quantum liquids and propose that $\TO^n$ can be obtained from $\TO^0$ by iterative deloopings (see Hypothesis \ref{hyp:to-cond}). Then we introduce the category $\TOsk^n$ of the topological skeletons of $n$D quantum liquids and compute it via iterative deloopings of $\TOsk^0$. This computation gives the main physical result of this work as summarized in Theorem \ref{cor:TOts-n}. We unravel this heavily loaded result in low dimensions in Example \ref{exam:QL1}, \ref{exam:QL2} and \ref{exam:QL3}, where we show that $\TOsk^2$ and $\TOsk^3$ reproduce the topological skeletons of all 2D rational CFT's and those of gapped/gapless boundaries of 3D topological orders obtained earlier in \cite{KZ20,KZ21,KYZ21} and those of gapless boundaries of 3D SPT/SET orders. In the end, we argue that $\TO^n\simeq \TOsk^n$. 

\medskip
The layout of this work is as follows. 
In Section \ref{sec:n-cat}, we briefly review some basic notions in higher categories and some formal facts about the condensation theory of $n$-categories from \cite{GJF19,JF20}, and compute the deloopings of certain coslice $n$-categories. In Section \ref{sec:separable}, we develop the theories of separable $n$-categories and $E_m$-multi-fusion $n$-categories. 
In Section \ref{sec:unitary}, we introduce the notion of $*$-condensation and that of a unitary (multi-fusion) $n$-category and compute the $*$-deloopings of certain coslice $n$-categories. In Section \ref{sec:TO}, we study the higher categories $\TO^n$ and $\TOsk^n$. 

\medskip
\noindent{\bf Acknowledgments}:
We would like to thank Xiao-Gang Wen for many inspiring discussions and long term collaboration. His persistence in pursuing the most fundamental questions inspired this work. We also thank Theo Johnson-Freyd, David Penneys, David Reutter, Hao Xu and Holiverse Yang for comments. We are supported by Guangdong Provincial Key Laboratory (Grant No.2019B121203002). LK is also supported by Guangdong Basic and Applied Basic Research Foundation under Grant No. 2020B1515120100 and by NSFC under Grant No. 11971219. HZ is also supported by NSFC under Grant No. 11871078 and by Startup Grant of Tsinghua University and BIMSA.

\section{The language of higher categories} \label{sec:n-cat}
In this section, we review some basic notions in higher categories and the notion of condensation completion introduced by Gaiotto and Johnson-Freyd \cite{GJF19}. We follow \cite{GJF19,JF20} to use $n$-category to mean a weak $n$-category without specifying a concrete model. See \cite{nlab} for a list of proposed definitions and references.

\subsection{$n$-Categories}

We use $\Cat_n$ to denote the $(n+1)$-category of $n$-categories and use $\Fun(\CC,\CD)$ to denote $\Hom_{\Cat_n}(\CC,\CD)$, the $n$-category of functors and (higher) natural transformations. For an $n$-category $\CC$, we use $\CC^{\op k}$ to denote the $n$-category obtained by reversing all the $k$-morphisms. Unless stated other wise, $\CC^\op$ means $\CC^{\op1}$.

\smallskip

A {\em monoidal} or {\em $E_1$-monoidal $n$-category} is a pair $(\CC,B\CC)$ where $B\CC$ is an $(n+1)$-category with a single object $\bullet$ and $\CC=\Hom_{B\CC}(\bullet,\bullet)$.
The identity 1-morphism $\Id_\bullet$ is referred to as the {\em tensor unit} of $\CC$ and denoted by $\one_\CC$. 
By induction on $m$, an {\em $E_m$-monoidal $n$-category} is pair $(\CC,B\CC)$ where $B\CC$ is an $E_{m-1}$-monoidal $(n+1)$-category with a single object $\bullet$ and $\CC=\Hom_{B\CC}(\bullet,\bullet)$. Note that an $E_m$-monoidal $n$-category consists of a finite series $(\CC,B\CC,B^2\CC,\dots,B^m\CC)$. By abusing terminology, we also refer to $\CC$ as an $E_m$-monoidal $n$-category. An $E_2$-monoidal $n$-category is also referred to as a {\em braided monoidal $n$-category}.

For an $E_m$-monoidal $n$-category $\CC$, we use $\CC^{\op k}$ where $k>-m$ to denote the $E_m$-monoidal $n$-category obtained by reversing all the $k$-morphisms, i.e. $B^m(\CC^{\op k}) = (B^m\CC)^{\op(k+m)}$. In particular, $\CC^{\op0}$ is denoted by $\CC^\rev$ and $\CC^{\op(-1)}$ is denoted by $\bar\CC$.

We say that an $n$-category $\CC$ {\em has duals}, if every $k$-morphism has both a left dual and a right dual for $1\le k<n$. We say that an $E_m$-monoidal $n$-category $\CC$ {\em has duals}, if the $(n+1)$-category $B\CC$ has duals or, equivalently, the $(n+m)$-category $B^m\CC$ has duals.

\smallskip

\void{
The $(n+1)$-category $E_0\Cat_n$ of {\em $E_0$-monoidal $n$-categories} is defined by the fiber product
$$\xymatrix{
  E_0\Cat_n \ar[r] \ar[d] & \Fun(\{0\to1\},\Cat_n) \ar[d]^{\ev_0} \\
  \{\bullet\} \ar[r]^\iota & \Cat_n \\
}
$$
where $\ev_0$ is the evaluation at $0$ and $\iota$ is the inclusion of a trivial category $\bullet$.
By definition, an $E_0$-monoidal $n$-category consists of a pair $(\CC,\one_\CC)$ where $\CC$ is an $n$-category (the value at 1) and $\one_\CC\in\CC$ is a distinguished object (picked out by the arrow $\bullet\to\CC$); an $E_0$-monoidal functor $(\CC,\one_\CC)\to(\CD,\one_\CD)$ between $E_0$-monoidal $n$-categories is a functor $F:\CC\to\CD$ such that $F(\one_\CC)=\one_\CD$; an $E_0$-monoidal (higher) natural transformation is a (higher) natural transformation that is trivial on the distinguished object.
We use $\Fun^{E_0}((\CC,\one_\CC),(\CD,\one_\CD))$ to denote $\Hom_{E_0\Cat_n}((\CC,\one_\CC),(\CD,\one_\CD))$ which is a subcategory of $\Fun(\CC,\CD)$.
}

An {\em $E_0$-monoidal $n$-category} is a pair $(\CC,\one_\CC)$ where $\CC$ is an $n$-category and $\one_\CC\in\CC$ is a distinguished object. An {\em $E_0$-monoidal functor} $(\CC,\one_\CC)\to(\CD,\one_\CD)$ between $E_0$-monoidal $n$-categories is a functor $F:\CC\to\CD$ such that $F(\one_\CC)=\one_\CD$.\footnote{Mathematicians might prefer to add an isomorphism $F(\one_\CC)\simeq\one_\CD$ as a defining data. This definition does not change the homotopy type of the space of $E_0$-monoidal functors. Since this paper is written for a physics journal, we prefer to keeping it simple.} An {\em $E_0$-monoidal (higher) natural transformation} is a (higher) natural transformation that is trivial on the distinguished object. We use $E_0\Cat_n$ to denote the $(n+1)$-category formed by the $E_0$-monoidal $n$-categories, $E_0$-monoidal functors and $E_0$-monoidal (higher) natural transformations and use $\Fun^{E_0}((\CC,\one_\CC),(\CD,\one_\CD))$ to denote $\Hom_{E_0\Cat_n}((\CC,\one_\CC),(\CD,\one_\CD))$ which is a subcategory\footnote{It is not a full subcategory.} of $\Fun(\CC,\CD)$.

Note that for an $E_m$-monoidal $n$-category $\CC$ where $m\ge1$, the iterated delooping $B^m\CC$, together with the distinguished object $\bullet$, defines an object of $E_0\Cat_{n+m}$. 
We use $E_m\Cat_n$ to denote the full subcategory of $E_0\Cat_{n+m}$ consisting of all the iterated deloopings $B^m\CC$ of $E_m$-monoidal $n$-categories. It is actually an $(n+1)$-category.
We use $\Fun^{E_m}(\CC,\CD)$ to denote $\Fun^{E_0}(B^m\CC,B^m\CD)$, the $n$-category formed by {\em $E_m$-monoidal functors} and {\em invertible $E_m$-monoidal (higher) natural transformations}. 

By definition, the delooping functor $B: E_m\Cat_n\to E_{m-1}\Cat_{n+1}$, $\CC\mapsto B\CC$ has a right adjoint $\CD \mapsto \Omega\CD := \Hom_\CD(\one_\CD,\one_\CD)$. That is, $\Fun^{E_{m-1}}(B\CC,\CD) \simeq \Fun^{E_m}(\CC,\Omega\CD)$.

We have an evident forgetful functor $E_{m+1}\Cat_n \to E_m\Cat_n$ for each $m\ge1$ and a forgetful functor $E_1\Cat_n \to E_0\Cat_n$, $\CC \mapsto (\CC,\one_\CC)$.
The $(n+1)$-category $E_\infty\Cat_n$ of {\em symmetric monoidal} or {\em $E_\infty$-monoidal $n$-categories} is defined to be the inverse limit of the system 
$$\cdots \to E_m\Cat_n \to E_{m-1}\Cat_n \to \cdots \to E_0\Cat_n.$$
By definition, a symmetric monoidal $n$-category consists of an infinite series $(\CC,B\CC,B^2\CC,\cdots)$. 

\smallskip

The $(n+1)$-categories $\Cat_n$ and $E_m\Cat_n$, $0\le m\le\infty$ are symmetric monoidal under Cartesian product. Moreover, $\Fun(\CA\times\CB,\CC) \simeq \Fun(\CA,\Fun(\CB,\CC))$ for $\CA,\CB,\CC\in\Cat_n$.


\smallskip

For a monoidal $n$-category $\CC$, the $(n+1)$-category $\LMod_\CC(\Cat_n)$ of {\em left $\CC$-modules} is defined to be $\Fun(B\CC,\Cat_n)$ and the $(n+1)$-category $\RMod_\CC(\Cat_n)$ of {\em right $\CC$-modules} is defined to be $\Fun(B\CC^\rev,\Cat_n)$. We use $\Fun_\CC(\CM,\CN)$ to denote $\Hom_{\LMod_\CC(\Cat_n)}(\CM,\CN)$.

For monoidal $n$-categories $\CC$ and $\CD$, the $(n+1)$-category $\BMod_{\CC|\CD}(\Cat_n)$ of {\em $\CC$-$\CD$-bimodules} is defined to be $\Fun(B\CC,\Fun(B\CD^\rev,\Cat_n))$, which is equivalent to $\LMod_{\CC\times\CD^\rev}(\Cat_n)$. We use $\Fun_{\CC|\CD}(\CM,\CN)$ to denote $\Hom_{\BMod_{\CC|\CD}(\Cat_n)}(\CM,\CN)$.

\subsection{Additive $n$-categories}

We formulate a definition of an additive $n$-category emphasizing on that the additivity is a property of an $n$-category rather than additional data.

We say that an $n$-category $\CC$ is {\em quasi-additive} if $\CC$ has a zero object and finite products as well as finite coproducts such that the canonical 1-morphism $X\coprod Y \to X\times Y$ is invertible for all $X,Y\in\CC$. The coproduct $X\coprod Y$ is also denoted by $X\oplus Y$, referred to as the {\em direct sum} of $X$ and $Y$. We say that an object $X\in\CC$ is {\em indecomposable} if it is neither zero nor a direct sum of two nonzero ones.

\begin{rem}
The canonical 1-morphism $X\coprod Y \to X\times Y$ is determined by $X\coprod Y \xrightarrow{\Id_X\coprod0} X\coprod0 \simeq X$ and $X\coprod Y \xrightarrow{0\coprod\Id_Y} 0\coprod Y \simeq Y$.
\end{rem}

If $\CC$ is a quasi-additive $n$-category, then $\Hom_\CC(X,Y)$ carries a binary operation defined by
$$f+g: X \to X\times X \xrightarrow{f\times g} Y\times Y \simeq Y\coprod Y \to Y$$
for 1-morphisms $f,g:X\to Y$.
It is a good exercise to show that compositions of 1-morphisms in $\CC$ distribute over this binary operation.


For $n=1$, the binary operation $(f,g)\mapsto f+g$ endows $\Hom_\CC(X,Y)$ with the structure of an additive monoid.
We say that a 1-category $\CC$ is {\em additive} if it is quasi-additive and the additive monoid $\Hom_\CC(X,Y)$ is an abelian group for any objects $X,Y\in\CC$.
By induction on $n$, we say that an $n$-category $\CC$ is {\em additive} if it is quasi-additive and $\Hom_\CC(X,Y)$ is additive for any objects $X,Y\in\CC$ and the canonical 2-morphisms $f\coprod g\to f+g \to f\times g$ are invertible for any 1-morphisms $f,g:X\to Y$.

\begin{rem}
The canonical 2-morphism $f+g \to f\times g$ is determined by $f+g \xrightarrow{\Id_f+0} f+0 \simeq f$ and $f+g \xrightarrow{0+\Id_g} 0+g \simeq g$, and similarly for $f\coprod g\to f+g$.
Note that the composition $f\coprod g\to f+g \to f\times g$ agrees with the canonical one. Therefore, in an additive $n$-category, the binary operation $+$ realizes $\oplus$.
\end{rem}

A functor $F:\CC\to\CD$ between two additive $n$-categories is {\em additive} if $F$ preserves finite products or, equivalently, finite coproducts.
We use $\Cat^+_n$ to denote the subcategory of $\Cat_n$ formed by additive $n$-categories and additive functors, and use $\Fun^+(\CC,\CD)$ to denote $\Hom_{\Cat^+_n}(\CC,\CD)$.
To make the definition consistent, we define the 1-category $\Cat^+_0$ of {\em additive 0-categories} to be the 1-category of abelian groups.

The $(n+1)$-category $\Cat^+_n$ is additive, and the direct sum $\CC\oplus\CD$ is the Cartesian product $\CC\times\CD$.
We leave this fact as an exercise to the reader.

\smallskip

An {\em additive monoidal $n$-category} is a pair $(\CC,B\CC)$ where $B\CC$ is an additive $(n+1)$-category such that all objects are finite direct sums of a single indecomposable object $\bullet$ and $\CC=\Hom_{B\CC}(\bullet,\bullet)$. {\em Additive $E_m$-monoidal $n$-categories} are similarly defined. We use $E_m\Cat^+_n$ to denote the $(n+1)$-category of additive $E_m$-monoidal $n$-categories and use $\Fun^{+E_m}(\CC,\CD)$ to denote $\Hom_{E_m\Cat^+_n}(\CC,\CD)$

For an additive monoidal $n$-category $\CC$, the additive $(n+1)$-category $\LMod_\CC(\Cat^+_n)$ of {\em additive left $\CC$-modules} is defined to be $\Fun^+(B\CC,\Cat^+_n)$ and we use $\Fun^+_\CC(\CM,\CN)$ to denote $\Hom_{\LMod_\CC(\Cat^+_n)}(\CM,\CN)$. Categories of additive right modules and bimodules are defined similarly.

\subsection{Linear $n$-categories}

Let $R$ be a commutative ring and view $R$ as an additive symmetric monoidal 0-category. 
What we expect for a linear higher category is that $B^n R$ is an $R$-linear $n$-category freely generated by a single object. This motivates the following definition.

The $(n+1)$-category $\Cat^R_n$ of {\em $R$-linear $n$-categories} is defined to be $\LMod_{B^n R}(\Cat^+_n) = \Fun^+(B^{n+1}R,\Cat^+_n)$. We use $\Fun_R(\CC,\CD)$ to denote $\Hom_{\Cat^R_n}(\CC,\CD)$, the $R$-linear $n$-category formed by {\em $R$-linear functors} and {\em $R$-linear (higher) natural transformations}.

By definition, evaluation at the distinguished object $\bullet\in B^n R$ induces an equivalence $\Fun_R(B^n R,\CC) \simeq \CC$ for any $R$-linear $n$-category $\CC$. That is, $B^n R$ is an $R$-linear $n$-category freely generated by a single object, as expected.

\begin{exam}
(1) $\Cat^R_0 = \Fun^+(B R,\Cat^+_0)$ is the 1-category of $R$-modules.

(2) $B^{n+1}\Z$ is an additive $(n+1)$-category freely generated by a single object (we leave it as an exercise to show this fact by using the adjunction $\Omega^n \simeq \Fun^{+E_n}(B^n\Z,-)$).
Hence $\Cat^\Z_n = \Fun^+(B^{n+1}\Z,\Cat^+_n) \simeq \Cat^+_n$.
\end{exam}

An {\em $R$-linear monoidal $n$-category} is a pair $(\CC,B\CC)$ where $B\CC$ is an $R$-linear $(n+1)$-category such that all objects are finite direct sums of a single indecomposable object $\bullet$ and $\CC=\Hom_{B\CC}(\bullet,\bullet)$. An {\em $R$-linear $E_m$-monoidal $n$-category} is similarly defined. We use $E_m\Cat^R_n$ to denote the $(n+1)$-category of $R$-linear $E_m$-monoidal $n$-categories.

For an $R$-linear monoidal $n$-category $\CC$, the $R$-linear $(n+1)$-category $\LMod_\CC(\Cat^R_n)$ of {\em $R$-linear left $\CC$-modules} is defined to be $\Fun_R(B\CC,\Cat^R_n)$ and we use $\Fun^R_\CC(\CM,\CN)$ to denote $\Hom_{\LMod_\CC(\Cat^R_n)}(\CM,\CN)$. Categories of $R$-linear right modules and bimodules are defined similarly.

\subsection{Condensations}

The condensation completion of higher categories plays a crucial role in this work. For 1-categories, it is just the usual Karoubi completion or idempotent completion. It was generalized to 2-categories by Carqueville and Runkel \cite{CR16} and by Douglas and Reutter \cite{DR18}, and was later generalized to $n$-categories by Gaiotto and Johnson-Freyd \cite{GJF19} (see also Remark \ref{rem:cc-physics}).

The results from \cite{GJF19} are used extensively in this work. 
We adopt the assumptions on the (co)limits theory of $n$-categories needed in \cite{GJF19}. We assume that readers are familiar with that paper. We only briefly recall some formal aspects of Gaiotto and Johnson-Freyd's construction in this subsection. 

\smallskip

Let $\KarCat_n$ denote the full subcategory of $\Cat_n$ formed by condensation-complete $n$-categories. The inclusion $\KarCat_n \hookrightarrow \Cat_n$ admits a left adjoint $\CC\mapsto \Kar(\CC)$, where $\Kar(\CC)$ is obtained by taking Karoubi envelope iteratively. 

By abstract nonsense, the construction $\CC\mapsto\Kar(\CC)$ maps additive $n$-categories to additive ones and maps $R$-linear $n$-categories to $R$-linear ones hence supplies left adjoint functors to the inclusions $\KarCat^+_n \hookrightarrow \Cat^+_n$ and $\KarCat^R_n \hookrightarrow \Cat^R_n$.

The looping construction $\Omega: E_{m-1}\KarCat_{n+1} \to E_m\KarCat_n$ then has a left adjoint given by 
$$\CC \mapsto \Sigma\CC := \Kar(B\CC)$$ 
and similarly for $E_m\KarCat^+_n$ and $E_m\KarCat^R_n$.

Note that evaluation at the distinguished object $\bullet\in \Sigma^n R$ induces an equivalence $\Fun_R(\Sigma^n R,\CC) \simeq \CC$ for any $\CC\in \KarCat^R_n$. That is, $\Sigma^n R$ is a condensation-complete $R$-linear $n$-category freely generated by a single object.

\begin{rem}
By virtue of our convention, $B\CC$ is an additive 1-category for any additive monoidal 0-category $\CC$ thus $\Sigma\CC$ is automatically an additive 1-category. For example, $\Sigma\C$ is the 1-category of finite-dimensional vector spaces over $\C$.
\end{rem}

A remarkable property of $\KarCat^R_n$ is that for any $\CC,\CD\in\KarCat^R_n$ there exist $\CC\boxtimes\CD\in\KarCat^R_n$ and an $R$-linear equivalence
$$\Fun_R(\CC\boxtimes\CD,-) \simeq \Fun_R(\CC,\Fun_R(\CD,-)).$$
Therefore, $\KarCat^R_n$ carries a symmetric monoidal structure with the tensor product $\boxtimes$ and the tensor unit $\Sigma^n R$.
Then the $(n+1)$-categories $E_m\KarCat^R_n$ are also symmetric monoidal under the same tensor product.

An explicit construction of the tensor product $\CC\boxtimes\CD$ is given as follows \cite{JF20}.
For $n=0$, $\CC\boxtimes\CD$ is simply the tensor product of $R$-modules $\CC\otimes_R\CD$.
Then, by induction on $n$, let $\CC\otimes\CD$ be the $n$-category whose objects are pairs $(X,Y)\in\CC\times\CD$ and $\Hom_{\CC\otimes\CD}((X,Y),(X',Y')) = \Hom_\CC(X,X')\boxtimes\Hom_\CD(Y,Y')$ and define $\CC\boxtimes\CD$ to be $\Kar(\CC\otimes\CD)$.

\begin{rem}
Form the construction we see that $\Omega(\CC\boxtimes\CD) = \Omega\CC\boxtimes\Omega\CD$ for $\CC,\CD\in E_0\KarCat^R_n$. Consequently, $\Sigma(\CC\boxtimes\CD) = \Sigma\CC\boxtimes\Sigma\CD$ for $\CC,\CD\in E_1\KarCat^R_n$.
\end{rem}

\begin{rem}
When $n=1$, the tensor product $\CC\boxtimes\CD$ is distinct from Deligne's tensor product whose definition involves a colimit-preserving condition unless either of $\CC$ or $\CD$ is semisimple.
\end{rem}

\begin{rem} \label{rem:sigma-rigid}
The first part of the proof of \cite[Theorem 4.1.1]{GJF19} implies the following result. If $\CC$ is a condensation-complete symmetric monoidal $n$-category then $\Sigma^m\CC$ is $m$-rigid, i.e. every object has a dual and every $k$-morphism has both a left dual and a right dual for $1\le k<m$. If, in addition, $\CC$ is $j$-rigid then $\Sigma^m\CC$ is $(j+m)$-rigid.
\end{rem}

\begin{rem} \label{rem:cc-physics}
The mathematical theory of condensation completion grows out of the study of defects in physics. We give a brief historical remark using unified terminologies\footnote{Condensation completion was called `orbifold completion' in \cite{CR16}, `idempotent completion' in \cite{DR18} and `Karoubi completion' or `condensation completion' in \cite{GJF19}. Condensation descendants was used in \cite{KW14} and was called `condensates' in \cite{GJF19}. The category containing of all defects, including condensation descendants, is called a maximal $\mathrm{BF}^{pre}$-category in \cite{KW14}.}. In 2012, in a mathematical study of TQFT with defects, Carqueville and Runkel briefly introduced condensation completion for 2-categories in \cite{CR16} as a process of completing a given 2-category of defects by their condensation descendants (or condensates). In a physical study of $n$D topological orders \cite{KW14}, defects and condensation descendants were studied from both microscopic and macroscopic perspectives, and condensation completion was discussed in some low dimensional examples, and the possibility of an algorithmic way of defining it was conjectured \cite[Conjecture 11]{KW14}. 
In 2018, Douglas and Reutter \cite{DR18} thoroughly developed the mathematical theory of condensation completion for 2-categories, found the correct definition of a multi-fusion 2-category and used it to give a state-sum construction of 4D TQFT's \cite{DR18}. In 2019, Gaiotto and Johnson-Freyd developed the mathematical theory of condensation completion for higher categories in \cite{GJF19}. The notion of a multi-fusion $n$-category was first introduced by Johnson-Freyd in \cite{JF20}, where some fundamental results on multi-fusion $n$-categories were proved. Moreover, the necessity of condensation completion in the categorical description of quantum phases is demanded by the boundary-bulk relation as shown in potentially anomalous 2D CFT's \cite{KZ18b,KZ20,KZ21}, 4D Dijkgraaf-Witten theories \cite{KTZ20} and $n$D SPT/SET orders \cite{KLWZZ20a}.  
\end{rem}

\subsection{(Co)slice $n$-categories}

In their formulation of (op)lax twisted extended TQFT's \cite{JFS17}, Johnson-Freyd and Scheimbauer constructed rigorously an (op)lax variant of arrow categories of which (co)slice categories are certain subcategories.

Let $\CC$ be an $n$-category containing an object $A$. The {\em slice $n$-category} $\CC/A$ is defined informally as follows. An object is a pair $(X,x)$ where $X\in\CC$ and $x:X\to A$ is a 1-morphism in $\CC$. A 1-morphism $(X,x)\to(Y,y)$ is a pair $(f,\phi)$ where $f:X\to Y$ is a 1-morphism in $\CC$ and $\phi:x\to y\circ f$ is a 2-morphism in $\CC$. Higher morphisms are defined similarly.

The {\em coslice $n$-category} $A/\CC$ is defined to be $(\CC^{\op1\cdots\op n}/A)^{\op1\cdots\op n}$. More precisely, an object of $A/\CC$ is a pair $(X,x)$ where $X\in\CC$ and $x:A\to X$ is a 1-morphism in $\CC$. A 1-morphism $(X,x)\to(Y,y)$ is a pair $(f,\phi)$ where $f:X\to Y$ is a 1-morphism in $\CC$ and $\phi:f\circ x\to y$ is a 2-morphism in $\CC$. Higher morphisms are defined similarly.

\smallskip

By abstract nonsense, if $\CC$ is an $E_m$-monoidal $n$-category then $\one_\CC/\CC$ is also an $E_m$-monoidal $n$-category.

\smallskip

Moreover, if $\CC$ is condensation-complete then $A/\CC$ is also condensation-complete. 
In fact, we have a pullback diagram in $\Cat_n$
$$\xymatrix{
  A/\CC \ar[r] \ar[d] & \Fun^\mathrm{oplax}(\{0\to1\},\CC) \ar[d]^{\ev_0} \\
  \bullet \ar[r]^A & \CC \\
}
$$
where $A$ is identified with a functor from a trivial category $\bullet$ to $\CC$ and $\Fun^\mathrm{oplax}(\{0\to1\},\CC)$ is the oplax variant of the arrow category of $\CC$ constructed in \cite{JFS17} under the notation $\CC^\to$.
Since a condensation is an absolute (co)limit, $\Fun^\mathrm{oplax}(\{0\to1\},\CC) \in \KarCat_n$ if $\CC\in\KarCat_n$. Since the full subcategory $\KarCat_n\subset\Cat_n$ is closed under limit as the inclusion has a left adjoint, $A/\CC\in\KarCat_n$ if $\CC\in\KarCat_n$.

\void{
We illustrate this fact by the following example. Let $(e,\epsilon)$ be a condensation monad over $(X,x)\in A/\CC$ and let us find a condensation $(f,\phi): (X,x)\condense(Y,y)$. 
First, the condensation monad $e$ over $X$ determines a condensation $X\condense Y$ so that $e$ factorizes as $X\xrightarrow{f}Y\xrightarrow{f'}X$. Moreover, the condensation $f\circ f'\condense\Id_Y$ incorporates a pair of 2-morphisms $f\circ f' \xrightarrow{\beta} \Id_Y \xrightarrow{\beta'} f\circ f'$. 
Then the composition 
$$f\circ x \xrightarrow{\beta'\circ\Id_{f\circ x}} f\circ f'\circ f\circ x \xrightarrow{
\Id_f\circ\epsilon} f\circ x$$
carries the structure of a condensation monad thus determines a condensation $\phi: f\circ x \condense y$.
Eventually, we obtain a 1-morphism $(f,\phi): (X,x)\to (Y,y)$ in $A/\CC$ which extends to a condensation by assembling more higher morphisms constructed in a similar way.
}

\begin{lem} \label{lem:slice-cond}
Let $\CC$ be a condensation-complete monoidal $n$-category such that every object of $\CC$ is a condensate of $\one_\CC$. Then every object $(X,x)\in\bullet/\Sigma\CC$ is a condensate of $(\bullet,\one_\CC)$.
\end{lem}

\begin{proof}
By the construction of $\Sigma\CC$ there is a condensation $f:\bullet\condense X$. That is, there exist a pair of 1-morphisms $\bullet \xrightarrow{f} X \xrightarrow{f'} \bullet$ as well as a condensation $\beta: f\circ f' \condense \Id_X$. Then we have condensations 
$$(\bullet,\one_\CC) \condense (\bullet,f'\circ f) \condense (X,f) \condense (X,x).$$
The first one is induced by an arbitrary $\one_\CC\condense f'\circ f$. The second one is $(f,\beta\circ\Id_f)$. The third one is induced by the composite condensation $f \xrightarrow{\Id_f\circ\gamma} f\circ f'\circ x \xrightarrow{\beta\circ\Id_x} x$ for an arbitrary $\gamma:\one_\CC\condense x\circ f'$. 
\end{proof}

\begin{prop} \label{prop:slice-sigma}
Let $\CC$ be a condensation-complete monoidal $n$-category such that every object of $\CC$ is a condensate of $\one_\CC$. We have a canonical equivalence $\Sigma(\CC/\one_\CC) \simeq \bullet/\Sigma\CC$.
\end{prop}

\begin{proof}
According to Lemma \ref{lem:slice-cond}, every object of $\bullet/\Sigma\CC$ is a condensate of $(\bullet,\one_\CC)$. Hence the canonical functor $\Sigma\Omega(\bullet/\Sigma\CC) \to \bullet/\Sigma\CC$ is an equivalence.
Moreover, $\CC/\one_\CC \simeq \Omega(\bullet/B\CC) = \Omega(\bullet/\Sigma\CC)$ canonically.
\end{proof}

The additive version of Proposition \ref{prop:slice-sigma} is proved similarly.

\begin{prop} \label{prop:slice-sigma-add}
Let $\CC$ be a condensation-complete additive monoidal $n$-category such that every object of $\CC$ is a condensate of a direct sum of $\one_\CC$. We have a canonical equivalence $\Sigma(\CC/\one_\CC) \simeq \bullet/\Sigma\CC$.
\end{prop}

\section{Separable higher categories} \label{sec:separable}
In this section, we further develop the theory of separable and multi-fusion $n$-categories based on the works \cite{DR18,GJF19,JF20} but from a slightly different perspective. Some notions are defined differently but proved to be compatible with \cite{JF20}. 

We work on the base field $\C$ but the results also apply to other fields with some minor exceptions. See Remark \ref{rem:char-k}.

\subsection{Separable $n$-categories}

Let $\Vect$ denote the symmetric monoidal 1-category of finite-dimensional vector spaces.
Let $n\Vect$ denote the symmetric monoidal $n$-category $\Sigma^{n-1}\Vect = \Sigma^n\C$ \cite{GJF19}.

Let $\Cat^\C_n$ denote the $(n+1)$-category of (additive) $\C$-linear $n$-categories and let $\KarCat^\C_n$ denote the full subcategory of condensation-complete $\C$-linear $n$-categories. 
By slightly abusing notation, we use $\Fun(\CC,\CD)$ to denote $\Hom_{\Cat^\C_n}(\CC,\CD)$.

\smallskip

The following theorem is proved in the same way as \cite[Corollary 4.2.3]{GJF19}.

\begin{thm} \label{thm:sigma-rmod}
Let $\CA$ be a condensation-complete $\C$-linear monoidal $n$-category. The functor $\Hom_{\Sigma\CA}(\bullet,-): \Sigma\CA \to \RMod_\CA(\KarCat^\C_n)$ is fully faithful where $\bullet\in \Sigma\CA$ is the distinguished object.
Moreover, the following conditions are equivalent for an object $\CM\in\RMod_\CA(\KarCat^\C_n)$:
\begin{enumerate}
\item $\CM$ belongs to the essential image.
\item The functor $\Fun_{\CA^\rev}(\CM,-): \RMod_\CA(\KarCat^\C_n) \to \KarCat^\C_n$ preserves colimits.
\item The evaluation functor $\Fun_{\CA^\rev}(\CM,\CA)\boxtimes\CM \to \CA$ exhibits the left $\CA$-module $\Fun_{\CA^\rev}(\CM,\CA)$ dual to $\CM$.
\item $\CM$ has a left dual in $\LMod_\CA(\KarCat^\C_n)$.
\end{enumerate}
\end{thm}

Applying the theorem to $n\Vect$ we obtain:

\begin{cor}
The functor $\Hom_{(n+1)\Vect}(\bullet,-): (n+1)\Vect \to \Cat^\C_n$ is fully faithful. 
\end{cor}


\begin{defn}
A {\em separable $n$-category} is a $\C$-linear $n$-category that lies in the essential image of the above functor.
\end{defn}

In what follows, we identify $(n+1)\Vect$ with the full subcategory of $\Cat^\C_n$ formed by the separable $n$-categories.

\begin{rem}
Since $(n+1)\Vect$ has duals by \cite[Theorem 4.1.1]{GJF19}, all separable $n$-categories have duals.
Since $(n+1)\Vect$ is essentially small, all separable $n$-categories are essentially small.
\end{rem}

\begin{rem}
According to Theorem \ref{thm:sigma-rmod}, the evaluation functor $\Fun(\CC,n\Vect)\boxtimes\CC \to n\Vect$ exhibits $\Fun(\CC,n\Vect)$ dual to $\CC$ for any separable $n$-category $\CC$.
\end{rem}

\begin{cor} \label{cor:sep-dual}
Let $\CC$ be a condensation-complete $\C$-linear $n$-category. The following conditions are equivalent:
\begin{enumerate}
\item $\CC$ is a separable $n$-category.
\item $\CC$ is fully dualizable in $\KarCat^\C_n$.
\item $\CC$ is 1-dualizable in $\KarCat^\C_n$.
\end{enumerate}
\end{cor}

\begin{proof}
$(1)\Rightarrow(2)$ is because $(n+1)\Vect$ has duals.
$(2)\Rightarrow(3)$ is trivial. 
$(3)\Rightarrow(1)$ Apply Theorem \ref{thm:sigma-rmod} to $n\Vect$.
\end{proof}

\begin{prop}
Let $\CC$ be a separable $n$-category. Then $\Hom_\CC(A,B)$ is a separable $(n-1)$-category for any two objects $A,B\in\CC$.
\end{prop}

\begin{proof}
We may assume $\CC=\Hom_{(n+1)\Vect}(\bullet,X)$ where $X\in(n+1)\Vect$. Then $\Hom_\CC(A,B)$ is canonically identified with $\Hom_{n\Vect}(\bullet,A^R\circ B)$ hence is a separable $(n-1)$-category.
\end{proof}

\begin{exam}
(1) $n\Vect = \Hom_{(n+1)\Vect}(\bullet,\bullet)$ is a separable $n$-category.

(2) A separable 0-category is a finite-dimensional vector space.

(3) Giving a condensation algebra in $\Vect$ is equivalent to giving a separable algebra. Giving a bimodule over condensation algebras is equivalent to giving a finite-dimensional bimodule over separable algebras.
Therefore, $2\Vect$ is equivalent to the symmetric monoidal 2-category of separable algebras, finite-dimensional bimodules and bimodule maps. In particular, a separable 1-category is precisely a finite semisimple 1-category.
\end{exam}

\begin{defn}
We say that an object $A$ of a separable $n$-category $\CC$ is {\em simple} if it is indecomposable, i.e. it is neither zero nor a direct sum of two nonzero objects. 
\end{defn}


\begin{prop}
Let $A$ be an object of a separable $n$-category $\CC$.
(1) $A$ is simple if and only if $\Id_A$ is a simple object of the separable $(n-1)$-category $\Hom_\CC(A,A)$.
(2) $A$ is a finite direct sum of simple objects.
\end{prop}

\begin{proof}
(1) If $A=A_1\oplus A_2$ then $\Id_A = e_1\oplus e_2$ where $e_i$ is the composition $A \to A_i\to A$. Conversely, if $\Id_A=e_1\oplus e_2$ then $e_i$ is an idempotent thus determines a condensation $A\condense A_i$ so that $A=A_1\oplus A_2$ by the uniqueness of condensation.

(2) The claim is clearly true for $n=1$. For $n>1$, $\Id_A$ is a finite direct sum of simple objects by the inductive hypothesis, so is $A$.
\end{proof}

\begin{prop} \label{prop:simp-condense}
Let $f:A\to B$ be a nonzero 1-morphism in a separable $n$-category $\CC$ where $B$ is simple. Then $f$ extends to a condensation $A\condense B$.
\end{prop}

\begin{proof}
For $n=1$, $f$ is a split surjection hence extends to a condensation. For $n>1$, the counit map $v:f\circ f^R\to\Id_B$ is a nonzero 1-morphism in $\Hom_\CC(B,B)$ where $\Id_B$ is simple. By the inductive hypothesis, $v$ extends to a condensation, as desired.
\end{proof}

\begin{cor} \label{cor:simp-nonzero}
Let $\CC$ be a separable $n$-category. 
(1) If $A\xrightarrow{f}B\xrightarrow{g}C$ are nonzero 1-morphisms between simple objects in $\CC$ then $g\circ f$ is nonzero.
(2) If $\CC$ is indecomposable then $\Hom_\CC(A,B)$ is nonzero for any simple objects $A,B\in\CC$.\footnote{
In other words, a separable $n$-category is indecomposable if and only if it is connected in the sense of Douglas and Reutter \cite{DR18} or, in physical language, connected by domain walls \cite{KW14}.}
\end{cor}

\begin{proof}
(1) $g\circ f$ extends to a condensation $A\condense C$ hence is nonzero.
(2) Let $\CA$ be the full subcategory of $\CC$ consisting of those objects $C$ satisfying $\Hom_\CC(A,C)=0$ and let $\CB$ be the full subcategory consisting of those objects $D$ satisfying $\Hom_\CC(D,C)=0$ for all $C\in\CA$. Then $\CC=\CA\oplus\CB$ by (1) hence $\CA=0$.
\end{proof}

\begin{cor} \label{cor:sigma-hom}
Let $\CC$ be an indecomposable separable $n$-category. 
Then $\CC=\Sigma\Hom_\CC(A,A)$ for any nonzero object $A\in\CC$.
\end{cor}

\begin{proof}
Since $\CC$ is indecomposable, there exists a nonzero 1-morphism $f:A\to B$ by Corollary \ref{cor:simp-nonzero}(2) hence a condensation $A\condense B$ by Proposition \ref{prop:simp-condense} for any simple object $B\in\CC$.
\end{proof}

\begin{lem}
If $\CC$ is a separable $n$-category, so is $\CC^{\op k}$.
\end{lem}

\begin{proof}
Since $(n-k)\Vect$ is symmetric monoidal, we have a canonical equivalence $(n-k)\Vect \simeq (n-k)\Vect^\rev$ inducing an equivalence $(n+1)\Vect \simeq (n+1)\Vect^{\op(k+1)}$, $\CC\mapsto\CC^{\op k}$.
\end{proof}

\begin{prop} \label{prop:sep-yoneda}
The Yoneda embedding $j: \CC^\op \hookrightarrow \Fun(\CC,n\Vect)$ is an equivalence for any separable $n$-category $\CC$. Therefore, all $\C$-linear functors $\CC\to n\Vect$ are representable and the pairing $\Hom_\CC(-,-): \CC^\op\boxtimes\CC\to n\Vect$ exhibits $\CC^\op$ dual to $\CC$.
\end{prop}

\begin{proof}
We may assume that $\CC$ is indecomposable. Then both $\CC^\op$ and $\Fun(\CC,n\Vect)=\CC^\vee$ are indecomposable separable $n$-categories. Hence $j$ is an equivalence by Corollary \ref{cor:sigma-hom}.
\end{proof}

\begin{rem} 
Corollary \ref{cor:sep-dual} is a special case of a more general result. See \cite[Corollary 4.2.3 and Corollary 4.2.4]{GJF19}. Some results of this subsection were alluded in \cite{JF20}.
\end{rem}

\subsection{Multi-fusion $n$-categories}

The notion of a multi-fusion 1-category was studied long ago (see, for example, \cite{DM82}), but the name was coined in \cite{ENO05}. That of a multi-fusion $n$-category was first introduced for $n=2$ by Douglas and Reutter \cite{DR18} and later for all $n$ by Johnson-Freyd \cite{JF20}. We give an alternative definition and prove its compatibility with that in \cite{JF20} (see Remark \ref{rem:MFC-compare-with-JF}). 

\begin{defn}
A {\em multi-fusion $n$-category} is a condensation-complete $\C$-linear monoidal $n$-category $\CA$ such that $\Sigma\CA$ is a separable $(n+1)$-category. We say that $\CA$ is {\em indecomposable} if $\Sigma\CA$ is indecomposable. A multi-fusion $n$-category with a simple tensor unit is also referred to as a {\em fusion $n$-category}. We adopt the convention that $\C$ is the only fusion 0-category.
\end{defn}

\begin{prop} \label{prop:sep-mfc}
For any object $A$ of a separable $n$-category $\CC$, $\Hom_\CC(A,A)$ is a multi-fusion $(n-1)$-category. In particular, $\Fun(\CC,\CC)$ is a multi-fusion $n$-category for any separable $n$-category $\CC$. 
\end{prop}

\begin{proof}
We may assume that $\CC$ is indecomposable then apply Corollary \ref{cor:sigma-hom}.
\end{proof}

\begin{defn}
Let $\CA$ be a multi-fusion $n$-categories. We say that a $\C$-linear right $\CA$-module $\CM$ is {\em separable} if $\CM$ is a separable $n$-category. We use $\RMod_\CA((n+1)\Vect)$ to denote the full subcategory of $\RMod_\CA(\KarCat^\C_n)$ formed by the separable right $A$-modules. The notions of a {\em separable left module} and a {\em separable bimodule} are defined similarly.
\end{defn}

\begin{prop} \label{prop:sigma-rmod}
Let $\CA$ be a multi-fusion $n$-category. The functor $\Hom_{\Sigma\CA}(\bullet,-): \Sigma\CA \to \RMod_\CA((n+1)\Vect)$ is an equivalence.
\end{prop}

\begin{proof}
According to Theorem \ref{thm:sigma-rmod}, we need to show that the functor is essentially surjective. For $\CM \in \RMod_\CA((n+1)\Vect)$, the condensation $\otimes: \CA\boxtimes\CA\condense\CA$ induces a condensation $\CM\boxtimes_\CA\CA\boxtimes\CA \condense \CM\boxtimes_\CA\CA$, i.e. $\CM\boxtimes\CA\condense\CM$. Since $(n+1)\Vect\boxtimes\Sigma\CA
\simeq\Sigma\CA$, $\CM\boxtimes\CA$ and hence $\CM$ belongs to the essential image.
\end{proof}

\begin{cor} \label{cor:sigma-bmod}
Let $\CA$ and $\CB$ be two multi-fusion $n$-categories. The functor $\Fun(\Sigma\CA,\Sigma\CB) \to \BMod_{\CA|\CB}((n+1)\Vect)$, $F \mapsto F(\CA)$ is an equivalence.
\end{cor}

\begin{proof}
$\Fun(\Sigma\CA,\Sigma\CB) \simeq (\Sigma\CA)^\vee\boxtimes\Sigma\CB \simeq \Sigma(\CA^\rev\boxtimes\CB) \simeq \RMod_{\CA^\rev\boxtimes\CB}((n+1)\Vect) \simeq \BMod_{\CA|\CB}((n+1)\Vect)$.
\end{proof}

\begin{thm} \label{thm:mfc-bim}
The construction $\CA\mapsto\Sigma\CA$ defines a symmetric monoidal equivalence 
$$\Mor_1^\mathrm{mf}(n\Vect) \simeq (n+1)\Vect$$ 
where $\Mor_1^\mathrm{mf}(n\Vect)$ is the symmetric monoidal $(n+1)$-category formed by multi-fusion $(n-1)$-categories and separable bimodules.
\end{thm}

\begin{proof}
The functor $\CA\mapsto\Sigma\CA$ is essentially surjective by Corollary \ref{cor:sigma-hom} and fully faithful by Corollary \ref{cor:sigma-bmod}.
\end{proof}

\begin{cor} \label{cor:mfc-dual}
Let $\CA$ be a condensation-complete $\C$-linear monoidal $n$-category. The following conditions are equivalent:
\begin{enumerate}
\item $\CA$ is a multi-fusion $n$-category.
\item $\CA$ is fully dualizable in $\Mor_1(\KarCat^\C_n)$.
\item $\CA$ is 2-dualizable in $\Mor_1(\KarCat^\C_n)$.
\end{enumerate}
\end{cor}

\begin{proof}
$(1)\Rightarrow(2)$ is due to Theorem \ref{thm:mfc-bim}. 
$(2)\Rightarrow(3)$ is trivial. 
$(3)\Rightarrow(1)$ Since $\CA$ is 2-dualizable, the $n\Vect$-$\CA\boxtimes\CA^\rev$-bimodule $\CA$ has a left dual thus determines a functor $u:\Sigma n\Vect\to\Sigma\CA\boxtimes\Sigma\CA^\rev$ by Theorem \ref{thm:sigma-rmod}. Similarly, the $\CA^\rev\boxtimes\CA$-$n\Vect$-bimodule $\CA$ has a left dual thus determines a functor $v:\Sigma\CA^\rev\boxtimes\Sigma\CA\to\Sigma n\Vect$. Then $u$ and $v$ exhibits $\Sigma\CA^\rev$ dual to $\Sigma\CA$. Therefore, $\Sigma\CA$ is a separable $(n+1)$-category by Corollary \ref{cor:sep-dual}.
\end{proof}

\begin{exam} \label{exam:mfc1}
It is clear that a multi-fusion 1-category is a finite semisimple monoidal 1-category with duals, i.e. a multi-fusion 1-category defined in \cite{ENO05}. The converse is also true because a multi-fusion 1-category defined in \cite{ENO05} is fully dualizable in $\Mor_1(\KarCat^\C_1)$ \cite{DSPS20}\footnote{The separability is automatic in characteristic zero by \cite{ENO05,DSPS20}.}. Then by \cite[Theorem 1.4.8 and 1.4.9]{DR18} and Proposition \ref{prop:sigma-rmod}, separable 2-categories are exactly semisimple 2-categories defined in \cite{DR18}.
\end{exam}

\begin{rem} \label{rem:MFC-compare-with-JF}
Corollary \ref{cor:mfc-dual} is essentially given by \cite[Theorem 1]{JF20}. According to Corollary \ref{cor:mfc-dual}, the definition of a multi-fusion $n$-category coincides with that in \cite{JF20}.
We conjecture that the definition of a multi-fusion 2-category is equivalent to that in \cite{DR18}.
\end{rem}

\begin{defn}\label{defn:internal-hom}
Let $\CA$ be a $\C$-linear monoidal $n$-category and $\CM$ be a $\C$-linear left $\CA$-module. The {\em internal hom} $[x,y]$ for $x,y\in\CM$, if exists, is defined to be the object of $\CA$ representing the functor $\Hom_\CM(-\otimes x,y): \CA^\op\to\Cat^\C_{n-1}$. That is, $\Hom_\CA(-,[x,y]) \simeq \Hom_\CM(-\otimes x,y)$. We say that $\CM$ is {\em enriched in $\CA$} if $[x,y]$ exists for all $x,y\in\CA$.
\end{defn}

\begin{prop} \label{prop:mfc-enrich}
If $\CA$ is a multi-fusion $n$-category then every separable left $\CA$-module is enriched in $\CA$.
\end{prop}
\begin{proof}
By Proposition \ref{prop:sep-yoneda}, every $\C$-linear functor $\CA^\op\to n\Vect$ is representable.
\end{proof}

\begin{prop} \label{prop:mfc-morita}
Let $\CA$ be an indecomposable multi-fusion $n$-category. Let $\one_\CA = \oplus_i e_i$ be the simple decomposition so that $\CA = \oplus_{i,j}\CA_{i j}$ as a separable $n$-category where $\CA_{i j}=e_i\otimes\CA\otimes e_j$.
(1) $\Sigma\CA = \Sigma\CA_{i i}$. In particular, $\CA_{i i}$ is a fusion $n$-category.
(2) The $\CA_{i i}$-$\CA_{k k}$-bimodule map $\CA_{i j}\boxtimes_{\CA_{j j}}\CA_{j k} \to \CA_{i k}$ in invertible.
(3) The $\CA_{i i}$-$\CA_{j j}$-bimodule $\CA_{i j}$ is inverse to $\CA_{j i}$.
\end{prop}

\begin{proof}
We have a decomposition of right $\CA$-modules $\CA = \oplus_i\CA_i$ where $\CA_i= e_i\otimes\CA$. In view of Proposition \ref{prop:sigma-rmod}, we identify $\Sigma\CA$ with $\RMod_\CA((n+1)\Vect)$. Then 
$\CA = \Hom_{\Sigma\CA}(\CA,\CA)$ implies $\CA_{i j} = \Hom_{\Sigma\CA}(\CA_j,\CA_i)$. Invoking Corollary \ref{cor:sigma-hom}, we obtain (1). 
The equivalence $\Sigma\CA_{j j} = \Sigma\CA_{k k}$ maps $\CA_{i j}$ to $\CA_{i k}$ and maps $\CA_{i j} \simeq \CA_{i j}\boxtimes_{\CA_{j j}}\CA_{j j}$ to $\CA_{i j}\boxtimes_{\CA_{j j}}\CA_{j k}$. We obtain (2).
(3) is a consequence of (2).
\end{proof}

\begin{defn}
An {\em $E_m$-multi-fusion $n$-category} is a condensation-complete $\C$-linear $E_m$-monoidal $n$-category $\CA$ such that $\Sigma^m\CA$ is a separable $(n+m)$-category. 
We say that $\CA$ is {\em connected} if $\CA$ is an indecomposable separable $n$-category.
An $E_m$-multi-fusion $n$-category with a simple tensor unit is also referred to as an {\em $E_m$-fusion $n$-category}. We adopt the convention that $\C$ is the only $E_m$-fusion 0-category.
\end{defn}

\begin{rem}
If $\CA$ is an $E_m$-multi-fusion $n$-category, so is $\CA^{\op k}$ because $\Sigma^m(\CA^{\op k})=(\Sigma^m\CA)^{\op(k+m)}$ is separable.
\end{rem}

\begin{rem} \label{rem:mfc-fusion}
If $\CA$ is an indecomposable $E_m$-multi-fusion $n$-category where $m\ge2$ then $\CA$ is of fusion type. In fact, in the notations of Proposition \ref{prop:mfc-morita}, we have $X\otimes Y \simeq Y\otimes X \in \CA_{i i}\cap\CA_{j j}$ for $X\in\CA_{i j}$ and $Y\in\CA_{j i}$. Hence $\CA_{i j}=0$ whenever $i\ne j$.
\end{rem}

\begin{rem}
It was proved in sketch in \cite{JF20} that the following conditions are equivalent for a condensation-complete $\C$-linear $E_m$-monoidal $n$-category $\CA$:
\begin{enumerate}
\item $\CA$ is an $E_m$-multi-fusion $n$-category.
\item $\CA$ is fully dualizable in $\Mor_m(\KarCat^\C_n)$.
\item $\CA$ is $(m+1)$-dualizable in $\Mor_m(\KarCat^\C_n)$.
\end{enumerate}
\end{rem}

\begin{prop} \label{prop:indecom-fus}
Let $\CA$ be a $\C$-linear monoidal $n$-category. Suppose that $\CA$ is an indecomposable separable $n$-category. Then $\CA$ is a fusion $n$-category.
\end{prop}

\begin{proof}
The claim is trivial for $n=0$. We assume $n\ge1$.
By Proposition \ref{prop:simp-condense}, the tensor product functor $\otimes:\CA\boxtimes\CA\to\CA$ extends to a condensation in $(n+1)\Vect$. We assume that the condensation is given by the consecutive counit maps $v_1:\otimes\circ\otimes^R \to \Id_\CA$, $v_2:v_1\circ v_1^R \to \Id_{\Id_\CA}$, etc. terminated by an identity $v_n\circ w=1$. Since $\one_\CA\otimes\one_\CA\simeq\one_\CA$ and since every object of $\CA$ is a condensate of $\one_\CA$, $\otimes$ is a simple morphism of $(n+1)\Vect$. Thus the counit maps $v_1,\dots,v_{n-1}$ are all simple. Thus $w$ is (essentially) a scalar inverse to $v_n$. 

Since $\CA=\Sigma\Omega\CA$ by Corollary \ref{cor:sigma-hom}, $\CA$ is 1-rigid by Theorem \ref{thm:sigma-rmod}(3). Therefore, the canonical condensation $\otimes:\CA\boxtimes\CA\condense\CA$ lifts to $\BMod_{\CA|\CA}((n+1)\Vect)$, inducing a condensation $-\boxtimes\CA\condense\Id_\CM$ where $\CM=\RMod_\CA((n+1)\Vect)$ and therefore extending $-\boxtimes\CA:(n+1)\Vect\to\CM$ to a condensation. This shows that $\CM$ is separable hence $\CA\simeq\Omega(\CM,\CA)$ is a fusion $n$-category.
\end{proof}

\begin{cor} \label{cor:indecom-emfus}
Let $\CA$ be a $\C$-linear $E_m$-monoidal $n$-category where $m\ge1$. Suppose that $\CA$ is an indecomposable separable $n$-category. Then $\CA$ is an $E_m$-fusion $n$-category.
\end{cor}
\begin{proof}
Apply Proposition \ref{prop:indecom-fus} for $m$ times. 
\end{proof}

\begin{thm} \label{thm:sep-mf}
Let $\CA$ be a condensation-complete $\C$-linear $E_m$-monoidal $n$-category where $m\ge1$. Suppose that $\Sigma\CA$ is a separable $(n+1)$-category. Then $\CA$ is an $E_m$-multi-fusion $n$-category.
\end{thm}
\begin{proof}
The claim is trivial for $m=1$. For $m\ge2$, we may assume that $\Sigma\CA$ is an indecomposable separable $(n+1)$-category. 
Invoking Corollary \ref{cor:indecom-emfus}, we conclude that $\Sigma\CA$ is an $E_{m-1}$-fusion $(n+1)$-category. That is, $\CA$ is an $E_m$-multi-fusion $n$-category. 
\end{proof}

\begin{cor} \label{cor:em-omega}
If $\CA$ is an $E_m$-multi-fusion $n$-category where $n\ge1$, then $\Omega\CA$ is an $E_{m+1}$-multi-fusion $(n-1)$-category.
\end{cor}
\begin{proof}
By Proposition \ref{prop:sep-mfc}, $\Sigma\Omega\CA$ is a separable $n$-category. Then Apply Theorem \ref{thm:sep-mf} to $\Omega\CA$.
\end{proof}

\begin{rem}
According to Theorem \ref{thm:sep-mf}, the notion of a braided or symmetric fusion 1-category agrees with the usual one.
\end{rem}

\begin{rem}
In the final version of \cite{JF20}, Johson-Freyd derived the $m=1$ case of Corollary\,\ref{cor:em-omega} from an interesting result, which says that a full rigid monoidal subcategory $\CA'$ of a multi-fusion $n$-category $\CA$ is also a multi-fusion $n$-category \cite[Proposition II.12]{JF20}. However, there is a gap in his proof of this result: it is not clear why the condensation $\CA'\boxtimes\CA'\to\CA'$ stops at a {\em nonzero} $n$-morphism. Even in the simplest case $n=1$, we do not know how to fix the gap without resort to the dimension theory of fusion 1-categories which is not available for higher fusion categories yet. We believe the claim is true but a proof appeals to considerable development on higher fusion categories.
\end{rem}

\subsection{Centers} \label{sec:centers}
In this subsection, we study higher centers and prove a prediction in \cite{KW14,KWZ15}. 
The following definition is standard. See \cite[Section 5.3]{Lur14}. 

\begin{defn} \label{defn:centers}
Let $\CA$ be a condensation-complete $\C$-linear $E_m$-monoidal $n$-category. The {\em $E_m$-center} of $\CA$ is the universal condensation-complete $\C$-linear $E_m$-monoidal $n$-category $\FZ_m(\CA)$ equipped with a unital action $F:\FZ_m(\CA)\boxtimes\CA\to\CA$, i.e. a $\C$-linear $E_m$-monoidal functor rendering the following diagram in $E_m\KarCat^\C_n$ commutative up to equivalence:
$$\xymatrix{
  & \FZ_m(\CA)\boxtimes\CA \ar[rd]^F \\
  \CA \ar[ru]^{\one_{\FZ_m(\CA)}\boxtimes\Id_\CA} \ar[rr]^{\Id_\CA} && \CA .
}
$$
\end{defn}

\begin{exam}
For $\CA\in E_0\KarCat^\C_n$, we have $\FZ_0(\CA) = \Fun(\CA,\CA)$ which is independent of the distinguished object $\one_\CA$. In fact, giving a unital action $\CB\boxtimes\CA\to\CA$ is equivalent to giving a $\C$-linear functor $\CB\to\Fun(\CA,\CA)$ that maps $\one_\CB$ to $\Id_\CA$.
By slightly abusing notation, we use $\FZ_0(\CA)$ to denote $\Fun(\CA,\CA)$ for $\CA\in\KarCat^\C_n$.
\end{exam}

\begin{rem}
By definition, $\FZ_m(\CA^{\op k}) = \FZ_m(\CA)^{\op k}$ for $\CA\in E_m\KarCat^\C_n$. 
\end{rem}

\begin{rem}
The composition $\FZ_m(\CA)\boxtimes\FZ_m(\CA)\boxtimes\CA \to \FZ_m(\CA)\boxtimes\CA \to \CA$ induces a $\C$-linear $E_m$-monoidal functor $\FZ_m(\CA)\boxtimes\FZ_m(\CA)\to\FZ_m(\CA)$, promoting $\FZ_m(\CA)$ to a $\C$-linear $E_{m+1}$-monoidal category.
\end{rem}

\begin{thm} \label{thm:mfc-center}
For $\CA\in E_m\KarCat^\C_n$, $\FZ_m(\CA) = \Omega^k\FZ_{m-k}(\Sigma^k\CA)$ where $0\le k\le m$. In particular, $\FZ_m(\CA) = \Omega^m\Fun(\Sigma^m\CA,\Sigma^m\CA)$.
\end{thm}

\begin{proof}
Since $\Sigma$ is left adjoint to $\Omega$, we have 
$$\Fun^{E_m}(\CB,\Omega^k\CC) \simeq \Fun^{E_{m-k}}(\Sigma^k\CB,\CC)$$ 
for $\CB\in E_m\KarCat^\C_n$ and $\CC\in E_{m-k}\KarCat^\C_{n+k}$. In particular, 
$$\Fun^{E_m}(\CB\boxtimes\CA,\CA) \simeq \Fun^{E_{m-k}}(\Sigma^k\CB\boxtimes\Sigma^k\CA,\Sigma^k\CA).$$ Therefore, 
$$\Fun^{E_m}(\CB,\FZ_m(\CA)) \simeq \Fun^{E_{m-k}}(\Sigma^k\CB,\FZ_{m-k}(\Sigma^k\CA)) \simeq \Fun^{E_m}(\CB,\Omega^k\FZ_{m-k}(\Sigma^k\CA)).$$
Hence $\FZ_m(\CA) = \Omega^k\FZ_{m-k}(\Sigma^k\CA)$.
\end{proof}

\begin{exam}
For a multi-fusion $n$-category $\CA$, $\FZ_1(\CA) = \Omega\Fun(\Sigma\CA,\Sigma\CA) \simeq \Fun_{\CA|\CA}(\CA,\CA)$. 
\end{exam}

\begin{cor}
If $\CA$ is an $E_m$-multi-fusion $n$-category then $\FZ_m(\CA)$ is an $E_{m+1}$-multi-fusion $n$-category. 
\end{cor}
\begin{proof}
Combine Corollary \ref{cor:em-omega} and Theorem \ref{thm:mfc-center}.
\end{proof}

\begin{cor}
Let $\CA$ and $\CB$ be $E_m$-multi-fusion $n$-categories. The canonical $\C$-linear $E_m$-monoidal functor $\FZ_m(\CA)\boxtimes\FZ_m(\CB)\boxtimes\CA\boxtimes\CB \to \CA\boxtimes\CB$ induces a $\C$-linear $E_{m+1}$-monoidal equivalence $\FZ_m(\CA)\boxtimes\FZ_m(\CB) \simeq \FZ_m(\CA\boxtimes\CB)$. 
\end{cor}
\begin{proof}
We have $\Fun(\CC,\CC)\boxtimes\Fun(\CD,\CD) \simeq \Fun(\CC\boxtimes\CD,\CC\boxtimes\CD)$ for separable $n$-categories $\CC$ and $\CD$ because $(n+1)\Vect\boxtimes(n+1)\Vect\simeq(n+1)\Vect$.
\end{proof}

\begin{defn}
We say that the $E_m$-center $\FZ_m(\CA)$ of $\CA\in E_m\KarCat^\C_n$ is {\em trivial} if the canonical $\C$-linear $E_{m+1}$-monoidal functor $n\Vect \to \FZ_m(\CA)$ is invertible.
\end{defn}

\begin{exam}
$\FZ_m(n\Vect) = \Omega^m\FZ_0((n+m)\Vect)$ is trivial.
\end{exam}

\begin{lem} \label{lem:box-simple}
Let $\CC,\CD$ be separable $n$-categories. $(1)$ If $X\in\CC$ and $Y\in\CD$ are simple objects, then $X\boxtimes Y\in\CC\boxtimes\CD$ is also simple. $(2)$ If $f:X\to X'$ and $g:Y\to Y'$ are simple morphisms of $\CC$ and $\CD$, respectively, then $f\boxtimes g:X\boxtimes Y\to X'\boxtimes Y'$ is also simple.
\end{lem}
\begin{proof}
(2) is a consequence of (1) and the equivalence $\Hom_{\CC\boxtimes\CD}(X\boxtimes Y,X'\boxtimes Y') \simeq \Hom_\CC(X,X')\boxtimes\Hom_\CD(Y,Y')$.
(1) is clear for $n=1$. By induction on $n\ge2$, $\Id_{X\boxtimes Y}$ is simple as $\Id_X$ and $\Id_Y$ are simple. Hence $X\boxtimes Y$ is simple.
\end{proof}

\begin{lem} \label{lem:mfc-z1}
If $\CA$ is an indecomposable multi-fusion $n$-category then $\Sigma\FZ_1(\CA) = \FZ_0(\Sigma\CA)$.
\end{lem}

\begin{proof}
Note that $\FZ_0(\Sigma\CA) = \Fun(\Sigma\CA,\Sigma\CA) \simeq (\Sigma\CA)^\vee\boxtimes\Sigma\CA$, which is connected by Lemma \ref{lem:box-simple}(1). Thus $\Sigma\Omega\FZ_0(\Sigma\CA) = \FZ_0(\Sigma\CA)$ by Corollary \ref{cor:sigma-hom}, where the left hand side is $\Sigma\FZ_1(\CA)$ by Theorem \ref{thm:mfc-center}.
\end{proof}

\begin{prop} \label{prop:center-square=1}
$(1)$ If $\CC$ is an indecomposable separable $n$-category then $\FZ_1(\Omega(\CC,A)) = \Omega\FZ_0(\CC)$ for any nonzero object $A\in\CC$. 

$(2)$ If $\CC$ is a nonzero separable $n$-category then $\FZ_1(\FZ_0(\CC))$ is trivial.

$(3)$ If $\CA$ is an indecomposable multi-fusion $n$-category then $\FZ_2(\FZ_1(\CA))$ is trivial.
\end{prop}

\begin{proof}
(1) Combine Theorem \ref{thm:mfc-center} and Corollary \ref{cor:sigma-hom}.

(2) Viewing $\CC$ as an object of $(n+1)\Vect$, we see that the $E_1$-center of $\Fun(\CC,\CC) = \Omega((n+1)\Vect,\CC)$ is $n\Vect$ by (1).

(3) Applying Theorem \ref{thm:mfc-center}, Lemma \ref{lem:mfc-z1} and (2), we obtain $\FZ_2(\FZ_1(\CA)) = \Omega\FZ_1(\Sigma\FZ_1(\CA)) = \Omega\FZ_1(\FZ_0(\Sigma\CA)) \simeq n\Vect$.
\end{proof}

\begin{rem}
The $m=2,k=1$ case of Theorem \ref{thm:mfc-center} appeared in \cite[IV.B]{JF20}; the $n=m=k=1$ case appeared in \cite[Proposition 3.27]{KLWZZ20a}. The $n=1$ case of Lemma \ref{lem:mfc-z1} appeared in \cite[Theorem$^{\mathrm{ph}}$ 3.28]{KLWZZ20a}. 
Proposition \ref{prop:center-square=1}(3) was predicted in \cite{KW14,KWZ15} and can be viewed as the mathematical formulation of the physical result: the bulk of a bulk is trivial. 
\end{rem}

\begin{rem} \label{rem:char-k}
All the results from this section apply to linear higher categories over an arbitrary field $k$ with some minor exceptions as follows. 
When $k$ is not separably closed, a fusion 0-category is defined to be a separable division $k$-algebra; Lemma \ref{lem:box-simple}, Lemma \ref{lem:mfc-z1} and Proposition \ref{prop:center-square=1}(3) fail. When $\chara k>0$, multi-fusion 1-categories correspond to separable multi-fusion 1-categories in the literature \cite{DSPS20}; finite semisimple 1-categories and semisimple 2-categories are not necessarily separable.
\end{rem}

\section{Unitary higher categories} \label{sec:unitary}
In this section, we outline a theory of unitary higher categories based on a $*$-version of condensation completion. 

\subsection{$*$-Condensations}

Let $\CC$ be an $n$-category equipped with an involution $*:\CC\to\CC^{\op m}$, where $1\le m\le n$, which fixes all the objects and all the $k$-morphisms for $k<m$. A {\em $*$-$n$-condensation} in $\CC$ is a $*$-equivariant condensation $\spadesuit_n\to\CC$, where the walking $n$-condensation $\spadesuit_n$ is endowed with the involution $*:\spadesuit_n\to\spadesuit_n^{\op m}$ that swaps the two generating $m$-morphisms and fixes all the objects and all the other generating morphisms.

By induction on $m$, we say that $\CC$ is {\em $*$-condensation-complete}, if every $*$-equivariant condensation monad in $\CC$ extends to a $*$-condensation and if $\Hom_\CC(X,Y)$ is $*$-condensation-complete when $m>1$ or condensation-complete when $m=1$ for all objects $X,Y\in\CC$.

In the special case $\CC=B\CD$ where $\CD$ is a $*$-condensation-complete monoidal $(n-1)$-category, we use $\Sigma_*\CD$ to denote the $*$-condensation completion of $\CC$. By construction, $\Sigma_*\CD$ inherits an involution $*: \Sigma_*\CD \to (\Sigma_*\CD)^{\op m}$.

\begin{rem} \label{rem:*-issue}
There are significant issues concerning the above definitions. Let us first consider the special case where $m=n$, the main concern of this work. Note that the associator $\alpha_{f,g,h}:(f\circ g)\circ h\to f\circ(g\circ h)$ for a sequence of $(n-1)$-morphisms $X\xrightarrow{f}Y\xrightarrow{g}Z\xrightarrow{h}W$ in $\CC$ induces the associator $\alpha_{f,g,h}^{-1}:(f\circ g)\circ h\to f\circ(g\circ h)$ in $\CC^{\op m}$. We should require the involution $*:\CC\to\CC^{\op n}$ maps $\alpha_{f,g,h}$ to $\alpha_{f,g,h}^{-1}$. That is, $\alpha_{f,g,h}$ is a {\em unitary} $n$-isomorphism of $\CC$. Informally speaking, the involution $*:\CC\to\CC^{\op m}$ is simply an involutive action on the class of $n$-morphisms of $\CC$ such that the coherence $n$-isomorphisms of $\CC$ are all unitary. Similarly, by a $*$-equivariant functor $F:\CC\to\CD$ we mean that $*\circ F=F\circ *$ on $n$-morphisms and that the coherence $n$-isomorphisms of $F$ are all unitary. In particular, a $*$-condensation in $\CC$ is a special case of a condensation and the $*$-condensation completion of $\CC$ might be realized as a subcategory of the condensation completion.

The case $m<n$ is similar but more complicated. Informally speaking, the involution $*:\CC\to\CC^{\op m}$ is an involutive action on the class of $\ge m$-morphisms of $\CC$ such that the $*$-action intertwines the coherence $\ge m$-equivalence of $\CC$ strictly. By a $*$-equivariant functor $F:\CC\to\CD$ we mean that $*\circ F=F\circ *$ on $\ge m$-morphisms and that the $*$-actions intertwine the coherence $\ge m$-equivalences of $F$ strictly. Again, a $*$-condensation in $\CC$ is a special case of a condensation and the $*$-condensation completion of $\CC$ might be realized as a subcategory of the condensation completion.

To summarize, the $*$-involution and the $*$-equivariance are defined in a strict rather than weak way. A comprehensive theory of $*$-condensations demands significant efforts and is far beyond the scope of this work. 
\end{rem}


\subsection{Unitary $n$-categories}

A {\em $*$-$n$-category} is a $\C$-linear $n$-category $\CC$ equipped with an anti-$\C$-linear involution $*:\CC\to\CC^{\op n}$ which fixes all the objects and all the $k$-morphisms for $k<n$. 
A {\em $*$-functor} $F:\CC\to\CD$ between two $*$-$n$-categories is a $*$-equivariant $\C$-linear functor. Similarly, a (higher) $*$-natural transformation is a $*$-equivariant $\C$-linear (higher) natural transformation. 

Let $\Cat^*_n$ denote the $(n+1)$-category formed by $*$-$n$-categories, $*$-functors and (higher) $*$-natural transformations and let $\KarCat^*_n$ denote the full subcategory of $*$-condensation-complete $*$-$n$-categories.
By slightly abusing notation, we use $\Fun(\CC,\CD)$ to denote $\Hom_{\Cat^*_n}(\CC,\CD)$.

\begin{rem}
As clarified in Remark \ref{rem:*-issue}, the coherence $n$-isomorphisms of a $*$-$n$-category $\CC$ are all unitary. This is exactly what we expect in a unitary theory.
\end{rem}

The theory of $*$-$n$-categories is completely parallel to that of $\C$-linear $n$-categories with some new features arising from the $*$-structure. In particular, all the results from the previous section have a $*$-version.

\smallskip

We use $\Hilb$ to denote the symmetric monoidal $*$-1-category of finite-dimensional Hilbert spaces, and
use $n\Hilb$ to denote the symmetric monoidal $*$-$n$-category $\Sigma_*^{n-1}\Hilb = \Sigma_*^n\C$.

\begin{prop}
The functor $\Hom_{(n+1)\Hilb}(\bullet,-): (n+1)\Hilb \to \Cat^*_n$ is fully faithful.
\end{prop}

\begin{defn}
A {\em unitary $n$-category} is a $*$-$n$-category that lies in the essential image of the above functor.
\end{defn}

\begin{prop}
Let $\CC$ be a unitary $n$-category. Then $\Hom_\CC(A,B)$ is a unitary $(n-1)$-category for any two objects $A,B\in\CC$.
\end{prop}

\begin{rem}
Since $n\Hilb$ is an iterated delooping of $\Hilb$, it is {\em positive}: $f^*\circ f\neq0$ for any nonzero $n$-morphism $f$. Since unitary $n$-categories are hom categories of $(n+1)\Hilb$, they are also positive.
\end{rem}

\begin{defn} \label{defn:Em-multi-fusion}
A {\em unitary $E_m$-multi-fusion $n$-category} is an $E_m$-monoidal $*$-$n$-category $\CA$ such that $\Sigma_*^m\CA$ is a unitary $(n+m)$-category. A unitary $E_m$-multi-fusion $n$-category with a simple tensor unit is also referred to as a {\em unitary $E_m$-fusion $n$-category}.
\end{defn}

\begin{exam}
(1) $n\Hilb$ is a unitary $n$-category.

(2) A unitary 0-category is a finite-dimensional Hilbert space.

(3) Giving a $*$-condensation algebra in $\Hilb$ is equivalent to giving a special $*$-Frobenius algebra. Giving a bimodule over $*$-condensation algebras is equivalent to giving a finite-dimensional $*$-bimodule over special $*$-Frobenius algebras.
Therefore, a unitary 1-category is $*$-equivalent to a finite direct sum of $\Hilb$.

(4) A unitary multi-fusion 1-category is a unitary multi-fusion 1-category in the usual sense and vice versa.
\end{exam}

\begin{rem}
Is a unitary $n$-category a separable $n$-category?
This is true for $n\le1$ by the above example. 
In view of Theorem \ref{thm:mfc-bim}, the question for $n=2$ is equivalent to whether finite semisimple modules over unitary multi-fusion 1-categories are unitarizable, which, as far as we know, remains open.

There are fusion 1-categories without unitary structure, for example, the Yang-Lee category of central charge $c=-\frac{22}5$ \cite{ENO05}. Thus $3\Hilb$ is not condensation-complete by Theorem \ref{thm:mfc-bim}. Therefore, $n\Hilb$ is not a separable $n$-category for $n\ge3$.
\end{rem}

\begin{lem}
Let $\CA$ be a monoidal $*$-$n$-category where $n\ge1$. Suppose that $\CA$ has duals and is a unitary $n$-category. Then $\CA$ is a unitary multi-fusion $n$-category.
\end{lem}

\begin{proof}
Since the tensor product functor $\otimes:\CA\boxtimes\CA\to\CA$ induces a nonzero functor from $\CA\boxtimes\CA$ to each simple direct summand of $\CA$, $\otimes$ extends to a $*$-condensation by the $*$-version of Proposition \ref{prop:simp-condense}. We assume that the $*$-condensation is given by the consecutive counit maps $v_1:\otimes\circ\otimes^R \to \Id_\CA$, $v_2:v_1\circ v_1^R \to \Id_{\Id_\CA}$, etc. terminated by an identity $v_n\circ v_n^*=1$.
Since $\CA$ is 1-rigid, the $*$-condensation $\otimes$ lifts to $\BMod_{\CA|\CA}((n+1)\Hilb)$, inducing a $*$-condensation $-\boxtimes\CA\condense\Id_\CM$ where $\CM=\RMod_\CA((n+1)\Hilb)$ and therefore extending $-\boxtimes\CA:(n+1)\Hilb\to\CM$ to a $*$-condensation. This shows that $\CM$ is a unitary $(n+1)$-category hence $\CA\simeq\Omega(\CM,\CA)$ is a unitary multi-fusion $n$-category.
\end{proof}

\begin{thm}
Let $\CA$ be an $E_m$-monoidal $*$-$n$-category where $n\ge1$. The following conditions are equivalent:
\begin{enumerate}
\item $\CA$ is a unitary $E_m$-multi-fusion $n$-category.
\item $\CA$ has duals and is a unitary $n$-category.
\end{enumerate}
\end{thm}

\begin{proof}
$(1)\Rightarrow(2)$ is clear. $(2)\Rightarrow(1)$ Apply the above lemma for $m$ times. 
\end{proof}



\subsection{Coslice construction}

\begin{exam} \label{exam:slice1}
The coslice 1-category $\C/\Hilb$ consists of the following data. An object $(X,x)$ consists of a finite-dimensional Hilbert space $X$ and a linear map $x:\C\to X$ (equivalently, a vector $x\in X$). A 1-morphism $f:(X,x)\to(Y,y)$ is a linear map $f:X\to Y$ such that $f\circ x=y$ (equivalently, $f(x)=y$).
\end{exam}

We have an involution $*':\Hilb\to\Hilb^\rev$ defined by $X\mapsto X^\vee$ on objects and by $f\mapsto f^{\vee*}$ on morphisms. It induces an involution $*':n\Hilb\to n\Hilb^{\op(n-1)}$. The delooping $\Sigma_{*'}n\Hilb$ obtained by using $*'$ is the same as $(n+1)\Hilb$ because every $*'$-equivariant condensation in $\Sigma_{*'}n\Hilb$ can be modified to be $*$-equivariant so that $\Sigma_{*'}n\Hilb \subset (n+1)\Hilb$ and every $*$-equivariant condensation in $(n+1)\Hilb$ can be modified to be $*'$-equivariant so that $\Sigma_{*'}n\Hilb \supset (n+1)\Hilb$.

We endow $\C/\Hilb$ with an involution by extending $*':\Hilb\to\Hilb^\rev$
$$*:\C/\Hilb\to(\C/\Hilb)^\rev, \quad (X,x)\mapsto(X^\vee,x^{\vee*}),\quad f\mapsto f^{\vee*}.$$

\begin{prop} \label{prop:hilb-slice}
We have canonical equivalences 
$$\Sigma_*^n(\C/\Hilb) \simeq \bullet/(n+1)\Hilb \simeq ((n+1)\Hilb/\bullet)^{\op(n+1)}, \quad \text{for even $n$},$$
$$\Sigma_*^n(\C/\Hilb) \simeq (n+1)\Hilb/\bullet \simeq (\bullet/(n+1)\Hilb)^{\op(n+1)}, \quad \text{for odd $n$}.$$ 
\end{prop}
\begin{proof}
Invoke the $*$-version of Proposition \ref{prop:slice-sigma-add} and use the equivalence 
$\C/\Hilb \simeq (\Hilb/\C)^\op$ defined by $(X,x)\mapsto(X,x^*)$, $f\mapsto f^*$. 
\end{proof}

\begin{exam} \label{exam:slice2}
The coslice 2-category $\bullet/2\Hilb$ consists of the following data. An object $(\CA,A)$ consists of a unitary 1-category $\CA$ and an object $A\in\CA$. A 1-morphism $(F,f):(\CA,A)\to(\CB,B)$ consists of a $*$-functor $F:\CA\to\CB$ and a 1-morphism $f:F(A)\to B$ in $\CB$. A 2-morphism $\xi:(F,f)\to(G,g)$ is a natural transformation $\xi:F\to G$ such that $f=g\circ\xi_A$.

The involution $*:\bullet/2\Hilb\to(\bullet/2\Hilb)^\op$ induced by that of $\Sigma_*(\C/\Hilb)$ fixes all the objects, maps a 1-morphism $(F,f):(\CA,A)\to(\CB,B)$ to $(F^\vee,f^{\vee*})$ where $f^\vee:A\to F^\vee(B)$ is the mate of $f:F(A)\to B$, and maps a 2-morphism $\xi$ to $\xi^{\vee*}$.
\end{exam}

\begin{rem} \label{rem:slice-indecomp}
The coslice $n$-category $\bullet/n\Hilb$ is not additive. We say that an object $(X,x)$ is {\em indecomposable} if $X$ is indecomposable and if $x$ is nonzero. Note that $X=\Sigma\Omega(X,x)$ if $(X,x)$ indecomposable. Therefore, giving an indecomposable object of $\bullet/n\Hilb$ is equivalent to giving an indecomposable unitary multi-fusion $(n-2)$-category.
\end{rem}

\section{Categories of quantum liquids} \label{sec:TO}
In this section, we study quantum liquids, a notion which unifies topological orders, SPT/SET orders and symmetry-breaking phases, 2D rational CFT's and additional gapless phases. We also use the mathematical tools developed in the previous sections to obtain some precise results of quantum liquids. All quantum phases are assumed to be anomaly-free unless we declare otherwise.

We first recall the theory of the boundaries of 3D topological orders then motivate the notion of a quantum liquid. We then recall and further develop a unified theory of SPT/SET orders and symmetry-breaking orders in all dimensions and explain that it is a part of a much bigger theory of quantum liquids. A quantum liquid can be completely characterized by two types of data: local quantum symmetry and topological skeleton. In this work, we focus on topological skeletons and compute the categories of topological skeletons as certain coslice categories. In the end, we discuss the relation between these categories and the categories of quantum liquids.

\subsection{Boundaries of 3D topological orders} \label{sec:twr}
In this subsection, we review the mathematical theory of gapped/gapless boundaries of 3D topological orders developed in \cite{KZ18b,KZ20,KZ21}. In particular, we review the idea of topological Wick rotation. 

\smallskip
A 3D topological order can be described by a pair $(\CC,c)$, where $\CC$ is a unitary modular tensor 1-category (UMTC) and $c$ is the chiral central charge. 
\begin{thm}[\cite{KZ18b,KZ20,KZ21}] 
A gapped/gapless boundary $\CX$ of the 3D topological order $(\CC,c)$ can be completely characterized by a triple $\CX=(V,\phi,\CP)$. 
\begin{enumerate}

\item $V$ has three different meanings representing the following three situations: 
\begin{itemize}
\item When $\CX$ is gapless and chiral, $V$ is called a {\it chiral symmetry} and is defined by a unitary rational vertex operator algebra (VOA) of central charge $c$. 
\item When $\CX$ is gapless and non-chiral, $V$ is called a {\it non-chiral symmetry} and is defined by a unitary rational full field algebra with chiral central charge $c^L-c^R=c$ \cite{HK07,KZ21}. 
\item When $c=0$ and the boundary $\CX$ is gapped,\footnote{The condition $c=0$ does not imply the boundary is gapped. For example, all gappable gapless boundaries (see for example \cite{CJKYZ20}) have the trivial chiral central charge $c=0$ and are necessarily non-chiral because the only unitary VOA $V$ with $c=0$ is the trivial one, i.e. $V=\mathbb{C}$. When $c\neq 0$, the boundary $\CX$ is necessarily gapless.} $V=\mathbb{C}$, which can be viewed as the trivial VOA or the trivial full field algebra. 
\end{itemize}
The category $\Mod_V$ of $V$-modules is a UMTC \cite{Hua08,KZ21}. 

\item $\CP$ is a unitary fusion right $\CC$-module \cite{KZ18}, i.e. a unitary fusion 1-category equipped with a braided functor $\CC\rightarrow \FZ_1(\CP)$, which is automatically a fully faithful embedding. We denote the centralizer of $\CC$ in $\FZ_1(\CP)$ by $\CC'_{\FZ_1(\CP)}$ and its time reversal by $\overline{\CC'_{\FZ_1(\CP)}}$ (i.e. reversing the braidings). 

\item $\phi: \Mod_V \to \overline{\CC'_{\FZ_1(\CP)}}$ is a braided equivalence between two UMTC's.

\end{enumerate}
\end{thm}

\begin{figure}
$$ 
\raisebox{-4em}{\begin{tikzpicture}
\fill[gray!30] (-3,0) rectangle (4,3) ;
\draw[very thick] (-1,3)--(-1,2.1) node[midway,left] {$M_{a,a}$};
\draw[very thick] (-1,1.9)--(-1,1.1) node[midway,left] {$M_{b,b}$};
\draw[very thick] (-1,0.9)--(-1,0) node[midway,left] {$M_{c,c}$};
\draw[fill=white] (-1.1,0.9) rectangle (-0.9,1.1) node[midway,right] {$\,M_{b,c}$} ;
\draw[fill=white] (-1.1,1.9) rectangle (-0.9,2.1) node[midway,right] {$\,M_{a,b}$} ;
\draw[very thick] (2,3)--(2,2.1) node[midway,right] {$M_{p,p}$};
\draw[very thick] (2,1.9)--(2,1.1) node[midway,right] {$M_{q,q}$};
\draw[very thick] (2,0.9)--(2,0) node[midway,right] {$M_{r,r}$};
\draw[fill=white] (1.9,0.9) rectangle (2.1,1.1) node[midway,left] {$M_{q,r}\,$} ;
\draw[fill=white] (1.9,1.9) rectangle (2.1,2.1) node[midway,left] {$M_{p,q}\,$} ;
\node at (-2.5,1.5) {$V$} ;
\node at (0.5,1.5) {$V$} ;
\node at (3.5,1.5) {$V$} ;
\end{tikzpicture}}
$$
\caption{macroscopic observables on the worldsheet of a 1+1D CFT 
}
\label{fig:observables}
\end{figure}

Such a triple characterizes all observables on the 2D world sheet as illustrated in Figure \ref{fig:observables}. In particular, $V$ is the chiral/non-chiral symmetry that is transparent on the entire 2D worldsheet; the objects $a,b,c,\cdots,p,q,r$ in $\CP$ are the labels of all topological defect lines (TDL) on the 2D worldsheet; $M_{a,a}$ is the space of fields living on the TDL labeled by $a$; and $M_{a,b}$ is the space of fields living at the 0D wall between TDL's labeled by $a$ and $b$. The labels of TDL's, together with all $M_{a,b}$, form an enriched category ${}^{\CC'_{\FZ_1(\CP)}}\CP$ with $\hom_{{}^{\CC'_{\FZ_1(\CP)}}\CP}(a,b):=M_{a,b}$, where ${}^{\CC'_{\FZ_1(\CP)}}\CP$ is defined by the canonical $\CC'_{\FZ_1(\CP)}$-action on $\CP$ via the so-called canonical construction \cite{MP17}. Note that both $\CP$ and $\CC'_{\FZ_1(\CP)}$ are abstract categories, so is ${}^{\CC'_{\FZ_1(\CP)}}\CP$, which does not have a direct connection to physical observables on the worldsheet until we supply the braided equivalence $\phi: \overline{\Mod_V} \to \CC'_{\FZ_1(\CP)}$. Through $\phi$, the hom spaces in ${}^{\CC'_{\FZ_1(\CP)}}\CP$ acquire their physical meanings as the physical observables in spacetime. In particular, the OPE of chiral or non-chiral fields on the 2D worldsheet and on the 1D TDL's are encoded in the compositions of the hom spaces in ${}^{\CC'_{\FZ_1(\CP)}}\CP$. Note that a different $\phi$ defines a different enrichment thus different macroscopic observables. For $\CX=(V,\phi,\CP)$, we introduce the following terminologies \cite{KZ20,KZ21}. 
\begin{enumerate}
\item the data $V$ and the braided equivalence $\phi: \overline{\Mod_V} \to \CC'_{\FZ_1(\CP)}$ are called the {\it local quantum symmetry} of $\CX$, which is denoted by $\CXlqs$; 


\item the data $\CP$, which can be viewed as an anomalous gapped boundary of $(\CC,c)$, is called the {\it topological skeleton} of $\CX$ and is denoted by $\CXsk$. 

\end{enumerate}
All together, we have $\CX=(\CXlqs,\CXsk)$. 

\begin{rem} \label{rem:anomalous-wall}
A boundary of a given nontrivial bulk phase is an example of an anomalous quantum phase. More generally, a domain wall $a$ between two quantum phases $X$ and $Y$, i.e. 
\begin{tikzpicture}[scale=0.6]
\draw[fill=white] (-0.2,-0.2) rectangle (0.2,0.2) ;
\draw (0,0) node{\scriptsize $a$};
\draw[thick] (-1,0) node[left] {\scriptsize $X$} -- (-0.2,0) ;
\draw[thick] (0.2,0) -- (1,0) node[right] {\scriptsize $Y$};


\end{tikzpicture}
, can be viewed as an anomalous boundary of $X$ with the anomaly defined by $Y$. Similarly, a defect junction of three quantum phases $X,Y$ and $Z$ can be viewed as an anomalous domain wall between $X$ and $Y$ with the anomaly defined by $Z$. 
\end{rem}

The boundary-bulk relation $\FZ_1({}^{\CC'_{\FZ_1(\CP)}}\CP) \simeq \CC$ holds \cite{KWZ17,KZ18,KYZZ21}. In other words, the gravitational anomaly of ${}^{\CC'_{\FZ_1(\CP)}}\CP$ is precisely given by $\CC$. 
We can also look at the same thing from a slightly different point of view. Note that $\CP$ can be viewed as an anomalous gapped boundary (recall Remark \ref{rem:anomalous-wall}) of the 3D topological order $(\CC,c)$, and its anomaly is encoded by the UMTC $\CC'_{\FZ_1(\CP)}$. There are two ways to cancel the anomaly of $\CP$. 
\begin{itemize}
\item An obvious way to cancel the anomaly is to attach the 3D topological order $(\overline{\CC'_{\FZ_1(\CP)}},c)$ to $\CP$ so that $\CP$ becomes an anomaly-free domain wall between the 3D topological orders $(\CC,c)$ and $(\overline{\CC'_{\FZ_1(\CP)}},c)$ as illustrated in Figure \ref{fig:TWR} with $\overline{\CB}:=\CC'_{\FZ_1(\CP)}$. As a consequence, the anomaly of $\CP$ is canceled in 2d space (or 3D spacetime). 
\item Another way to cancel the anomaly is to apply the so-called {\it topological Wick rotation}, which was introduced in \cite{KZ20}. More precisely, one can  `rotate' the 3D topological order $(\overline{\CC'_{\FZ_1(\CP)}},c)$ to the time direction as illustrated in Figure \ref{fig:TWR}. Physically, this `rotation' amounts to replace the physical observables in $\hom_\CP(a,b)$ by those in $\hom_{{}^{\CC'_{\FZ_1(\CP)}}\CP}(a,b)$ for $a,b\in\CP$. This replacement fixes the anomaly  in 2D spacetime. 
\end{itemize}
These two ways are connected by topological Wick rotation, which leads to many deep predictions as we will show later but remains a mystery to us.

\begin{figure}[htp]
$$ 
\raisebox{-2em}{\begin{tikzpicture}
\draw[thick] (-1,0)--(0.5,1.5) node[midway,left] {} ;
\draw[thick] (0,0)--(1.5,1.5) node[midway,left] {} ;
\draw[thick] (2,0)--(3.5,1.5) node[midway,left] {} ;
\draw[thick] (-1,0)--(2,0) node[midway,left] {} ;
\draw[thick] (0.5,1.5)--(3.5,1.5) node[midway,left] {} ;

\node at (-0.3,0.2) {\scriptsize $(\CB,c)$} ;
\node at (1,0.7) {\scriptsize $\CP$} ;
\node at (2.3,1) {\scriptsize $(\CC,c)$} ;
\end{tikzpicture}}
\xrightarrow{\mbox{\scriptsize topological Wick rotation}}
\raisebox{-2.5em}{\begin{tikzpicture}
\draw[thick] (0,0)--(0,0.8) node[midway,left] {} ;
\draw[thick] (0,0.8)--(1.5,2.3) node[midway,left] {} ;
\draw[thick] (1.5,1.5)--(1.5,2.3) node[midway,left] {} ;
\draw[thick] (0,0)--(1.5,1.5) node[midway,left] {} ;
\draw[thick] (2,0)--(3.5,1.5) node[midway,left] {} ;
\draw[thick] (0,0)--(2,0) node[midway,left] {} ;
\draw[thick] (1.5,1.5)--(3.5,1.5) node[midway,left] {} ;
\node at (0.9,1.3) {\scriptsize ${}^{\overline{\CB}}\CP$} ;
\node at (2.3,1) {\scriptsize $(\CC,c)$} ;
\end{tikzpicture}}
$$
\caption{the idea of topological Wick rotation
}
\label{fig:TWR}
\end{figure}



A special case is especially important. Consider the special case $(\CC,c)=(\Hilb,0)$ in Figure \ref{fig:TWR}. In this case, the topological Wick rotation gives an exact holographic duality between the 3D topological order $(\CB,c)$ with its gapped boundary $\CP$ (i.e. $\overline{\CB}\simeq \FZ_1(\CP)$) and an anomaly-free non-chiral 2D CFT $(V,\phi,\CP)$. While AdS/CFT correspondence is a duality between two gapless systems, topological Wick rotation defines a holographic duality between gapped and potentially gapless phases. Similar to AdS/CFT dictionary, this holographic duality also provides a long dictionary. For example, when $\CP$ is a modular tensor category, the Lagrangian algebra in $\CB$ associated to the boundary $\CP$ \cite{K14}, after the rotation, becomes precisely a modular-invariant 2D CFT \cite{KR09} (see \cite{BCDP22} for a related and more recent discussion). Applying topological Wick rotation to 3D topological orders with defects in all codimensions, we obtain a generalization of the exact holographic duality to defects \cite{KZ20,KZ21,KYZ21}. 

\begin{hyp} \label{hyp:quantum-liquid}
The following hypotheses, originally proposed in \cite{KZ20,KZ21}, are made more explicit. 
\begin{enumerate}

\item Topological Wick rotation works in all dimensions (see Figure\,\ref{fig:TWR_2}). When $\CC$ is trivial, it provides an exact duality between an $n+$1D topological order $\CB$ with a gapped boundary $\CP$ and a (potentially gapless) $n$D anomaly-free quantum phases. When $\CC$ is non-trivial, it provides an exact duality between $n$D anomalous gapped boundaries of an $n+$1D topological order $\CC$ (recall Remark\,\ref{rem:anomalous-wall}), i.e. the pair $(\CP,\CB)$, and $n$D anomaly-free (potentially gapless) boundaries of $\CC$, which can be viewed as an $n$D anomalous quantum phases. This exact duality generalizes to defects of all codimensions. We name all the $n$D (anomaly-free/anomalous) quantum phases obtained after TWR (anomaly-free/anomalous) {\it quantum liquids}. 

\item The description $\CX=(\CXlqs,\CXsk)$ works for all quantum liquids. The topological skeletons of quantum liquids (or boundaries) can all be obtained from topological Wick rotation (see Figure\,\ref{fig:TWR_2}). 

\item Boundary-bulk relation \cite{KWZ15,KWZ17} holds for all quantum liquids (see more discussion in Subsection \ref{sec:cat-QL}). 


\end{enumerate}
\end{hyp}

\begin{figure}
$$ 
\begin{array}{ccc}
\raisebox{-1em}{\begin{tikzpicture}[scale=0.6]
\draw[ultra thick] (0,0)--(-3,0) node[midway,above] {\scriptsize $\CC$} node[below] {\scriptsize $n$-th spatial direction} ; 
\draw[ultra thick] (2,0)--(0,0) node[midway,above] {\scriptsize $\CB$} ;
\draw[fill=white] (-0.1,-0.1) rectangle (0.1,0.1) node[midway,above] {\scriptsize $\CP$} ;
\end{tikzpicture}}
& \xrightarrow{\mathrm{TWR}} & 
\raisebox{-1em}{\begin{tikzpicture}[scale=0.6]
\draw[ultra thick] (0,2)--(0,0) node[midway,right] {\scriptsize $\overline{\CB}$} ;
\node[above] at (-2,2) {\scriptsize the time direction};
\draw[ultra thick] (0,0)--(-3,0) node[midway,above] {\scriptsize $\CC$} ; 
\draw[fill=white] (-0.1,-0.1) rectangle (0.1,0.1) node[midway,below] {\scriptsize $\CP$} ;
\end{tikzpicture}} 
=
\begin{tikzpicture}[scale=0.6]
\draw[ultra thick] (0,0)--(-3,0) node[midway,above] {\scriptsize $\CC$} ; 
\draw[fill=white] (-0.1,-0.1) rectangle (0.1,0.1) node[midway,above] {\scriptsize ${}^{\overline{\CB}}\CP$} ;
\end{tikzpicture}
\\
\mbox{\scriptsize  $\CP$ is a gapped domain wall beween} && \mbox{\scriptsize $n$D potentially gapless boundaries of}  \\
\mbox{\scriptsize  two $n+$1D topological orders $\CC$ and $\CB$} & & 
\mbox{\scriptsize the $n+$1D topological order $\CC$}
\end{array}  
$$
\caption{The idea of topological Wick rotation (TWR): Before TWR, slightly abusing notations, $\CP$ also denotes the category of all topological defects on the gapped wall and $\CB$ also denotes that of topological defects of codimension 2 and higher in the $n+$1D topological order $\CB$. After TWR, they become the category of topological sectors of states and that of operators, respectively (see \cite{KWZ21} for more details). The enriched higher category ${}^{\overline{\CB}}\CP$ summarizes all topological observables (or the topological skeleton) in the $n$D spacetime of the quantum liquid.}
\label{fig:TWR_2}
\end{figure}

A gapped boundary of the trivial $n+$1D topological order is an anomaly-free $n$D topological order $\CX$. In this case, the quantum liquid produced by the topological Wick rotation is precisely the topological order $\CX$. In other words, quantum liquid automatically include all topological orders. We explain in Subsection \ref{sec:spt-set} that quantum liquids also include all SPT/SET orders and symmetry-breaking orders, and, therefore, all `gapped quantum liquids' defined in \cite{ZW15,SM16}. This fact justifies the name `quantum liquid' (see Remark \ref{rem:the-name}).


\subsection{SPT/SET and symmetry-breaking orders} \label{sec:spt-set}
Using the main result in \cite{KLWZZ20a} and topological Wick rotation, we develop a unified theory of SPT/SET orders and symmetry-breaking orders as a part of the theory of quantum liquids. 

\smallskip
Let $\CR$ be a unitary symmetric fusion $n$-category. For example, $\CR=n\Rep G$ where $G$ is a finite group describing an onsite symmetry. 

\begin{thm}[\cite{KLWZZ20a}] \label{thm:spt-set}
An $n+1$D SPT/SET order with the higher symmetry $\CR$ can be completely characterized (up to invertible topological orders) by the following data: 
\begin{enumerate}
\item a unitary fusion $n$-category $\CA$ equipped with a braided faithful embedding $\CR \hookrightarrow \FZ_1(\CA)$ satisfying the following condition: 
\smallskip
\begin{enumerate}
\item[$(**)$] the composed functor $\CR \hookrightarrow \FZ_1(\CA) \to \CA$ is faithful;
\end{enumerate}
\item a braided equivalence $\phi: \FZ_1(\CR)\to \FZ_1(\CA)$ preserving the symmetry charges, i.e. rendering the following diagram commutative. 
\begin{equation} \label{eq:preserve-R}
\raisebox{1em}{\xymatrix@R=1em{
& \CR \ar@{^{(}->}[dl]_{} \ar@{^{(}->}[dr]^{} & \\
\FZ_1(\CR) \ar[rr]^{\phi} & & \FZ_1(\CA)
}}
\end{equation}
\end{enumerate}
\end{thm}

\begin{thm} \label{thm:sb-order}
When the condition $(**)$ and the preserving-charge condition (\ref{eq:preserve-R}) in Theorem \ref{thm:spt-set} do not hold, we obtain a characterization of a symmetry-breaking order. Altogether, we obtain the mathematical characterization (or classification) of all `gapped quantum liquids' (in the sense of \cite{ZW15,SM16} and up to invertible topological orders).\footnote{Theorem\,\ref{thm:spt-set} and \ref{thm:sb-order} can be stated differently as a characterization or classification of gapped liquids with algebraic higher symmetries or fusion category symmetries as in \cite{KLWZZ20b}.} 
\end{thm}


Note that $\CA$ in Theorem \ref{thm:spt-set} and \ref{thm:sb-order} can be viewed as an anomalous $n+$1D topological order \cite{KW14,KWZ15}. The anomaly is canceled along the space-direction by the $n+$2D bulk. However, this characterization cannot be a physical characterization of an $n+$1D SPT/SET or symmetry-breaking order because its $n+$2D bulk in a lattice model realization is completely empty. In a physical realization, the anomaly of $\CA$ is necessarily fixed in the $n+$1D spacetime instead of its $n+$2D bulk. We have two important observations. 
\begin{enumerate}
\item The macroscopic observables depicted in Figure \ref{fig:observables} or its higher dimensional analogues always work regardless the phase is gapped or gapless. In particular, the labels of topological defects of codimension $\geq 1$ form a monoidal higher category $\CP$. The macroscopic observables on defects are encoded in the internal homs of (higher) morphisms in $\CP$ thus form an enriched higher category, which will be defined precisely elsewhere. This is entirely similar to 2D rational CFT's. 
\item Onsite symmetries should be special cases of local quantum symmetries. 
\end{enumerate}
By these two observations, we propose to apply topological Wick rotation to Theorem \ref{thm:spt-set} and \ref{thm:sb-order}.  As a consequence, we obtain a physical characterization of an SPT/SET (or symmetry-breaking) order in terms of an enriched fusion higher category ${}^{\FZ_1(\CA)}\CA$, which will be defined elsewhere. By boundary-bulk relation \cite{KWZ15,KWZ17}, we expect that ${}^{\FZ_1(\CA)}\CA$ is anomaly-free in the sense that $\FZ_1({}^{\FZ_1(\CA)}\CA)\simeq n\Hilb$, which was rigorously proved only in the $n=1$ cases \cite{KZ18,KYZZ21}. The rigorous study of the $n>1$ cases are beyond this work and will be studied elsewhere. Now we assume that $\FZ_1({}^{\FZ_1(\CA)}\CA)\simeq n\Hilb$ is true for all $n$. In other words, the anomaly of $\CA$ is canceled by observables in $n+$1D spacetime. Similar to 2D CFT's, using a braided equivalence $\phi: \FZ_1(\CR) \to \FZ_1(\CA)$, we obtain an enriched higher category ${}^{\FZ_1(\CR)}\CA$, which should be viewed as the precise mathematical description of the spacetime observables in an $n+$1D SPT/SET (or symmetry-breaking) order. When $\CR=n\Rep G$, 
similar to AdS/CFT correspondence, the topological Wick rotation maps the 
gauge symmetry $G$ of the $n+$2D phase to the global symmetry $G$ of the 
$n+$1D phase. It is reasonable that this gauge/global interpretation of $G$ 
should generalize to all $\CR$.

Note that our proposal does not change the classification theory of SPT/SET orders in Theorem \ref{thm:spt-set}, which is consistent with many old classification results in physics. As we show in Example \ref{exam:2D-classification}, the 2D result in Theorem \ref{thm:sb-order} reproduces exactly the same classification of 2D gapped quantum liquids obtained earlier in \cite{CGW10b,SPGC11}, which was based on a different and microscopic approach. This fact provides a strong support of our proposal. After the appearance of this work in arXiv, our proposal was proved explicitly in the Ising chain and the Kitaev chain \cite{KWZ21} and in more general 2D models \cite{XZ22}. More precisely, the enriched fusion categories obtained from topological Wick rotation indeed give the physical characterizations of all the macroscopic observables in all phases realized in these 2D lattice models. We will provide more evidences in the future.

\begin{rem} \label{rem:the-name}
Above proposal also says that `gapped quantum liquids' defined in \cite{ZW15,SM16}, including topological orders, SPT/SET orders and symmetry-breaking orders, are examples of `quantum liquids' in our sense. We believe that two names are compatible in the sense that `gapped quantum liquids' in \cite{ZW15,SM16} are precisely gapped `quantum liquids' in our sense. 
\end{rem}

\begin{exam}
When $\CA=\CR$, the enriched higher category ${}^{\FZ_1(\CA)}\CA$ describes the physical observables in the spacetime of the $n+1$D SPT order with the symmetry $\CR$. The braided equivalences $\phi: \FZ_1(\CR) \to \FZ_1(\CR)$ satisfying the condition (\ref{eq:preserve-R}) form the group of $n+1$D SPT orders with the symmetry $\CR$. When $\CR=n\Rep G$ for a finite group $G$, 
the fusion $n$-category $n\Hilb_G = n\Hilb\times G$ is Morita equivalent to $n\Rep G$. A given Morita equivalence produces a braided equivalence $\phi: \FZ_1(\CR) \to \FZ_1(n\Hilb_G)$ and then an enriched higher category ${}^{\FZ_1(\CR)}n\Hilb_G$, which describes a phase with the symmetry $\CR$ completely broken. When $G=\Z_2$ and $n=1$, the trivial SPT order ${}^{\FZ_1(\Rep \Z_2)}\Rep \Z_2$ and the symmetry-breaking order ${}^{\FZ_1(\Rep \Z_2)}\Hilb_{\Z_2}$ were both realized in the Ising chain \cite{KWZ21}. For $n>1$, we predict that the enriched higher categories ${}^{\FZ_1(n\Rep\Z_2)}n\Rep\Z_2$ and ${}^{\FZ_1(n\Rep\Z_2)}n\Hilb_{\Z_2}$ can be realized in higher dimensional Ising models. 
\end{exam}

\begin{exam} \label{exam:2D-classification}
In 2D, when $\CR=\Rep G$ for a finite group $G$, by \cite{KZ18,KYZZ21,D10}, all such enriched fusion categories ${}^{\FZ_1(\CR)}\CA$ are of the form ${}^{\FZ_1(\CR)}(\FZ_1(\CR))_{A_{(H,\omega)}}$, where $H$ is a subgroup of $G$, $\omega \in H^2(H,U(1))$, $A_{(H,\omega)}$ is the Lagrangian algebra in $\FZ_1(\CR)$ determined by $(H,\omega)$ \cite{D10} and $(\FZ_1(\CR))_{A_{(H,\omega)}}$ denotes the category of right $A_{(H,\omega)}$-modules in $\FZ_1(\CR)$. The enrichment is defined by the canonical $\FZ_1(\CR)$-action on $(\FZ_1(\CR))_{A_{(H,\omega)}}$. This result recovers the well-known classification of all 2D gapped quantum phases with bosonic finite onsite symmetries by the triples $(G,H,\omega)$ \cite{CGW10b,SPGC11}, a result which was based on a microscopic definition of gapped quantum phases. Moreover, one can easily obtain the classification of the 1D boundaries of these 2D phases by classifying all closed modules \cite{KZ21} over enriched fusion 1-category ${}^{\FZ_1(\CR)}(\FZ_1(\CR))_{A_{(H,\omega)}}$, or equivalently, by classifying all modules over $(\FZ_1(\CR))_{A_{(H,\omega)}}$. When $G$ is abelian, the physical meanings of above enriched fusion categories and their closed modules were proved in 1D lattice models in \cite{XZ22}. It turns out that the same enriched fusion 1-categories (resp. triples) characterize (resp. classify) all 2D gapped quantum phases with fermionic finite onsite symmetries \cite{KWZ21}. The classification part of this result was previously conjectured via a model-dependent argument based on Jordan-Wigner transformations. 
\end{exam}

We refer to the $\CA$ in Theorem \ref{thm:spt-set} and \ref{thm:sb-order}, viewed as an $n+$1D anomalous topological order, as the topological skeleton of the $n+1$D gapped quantum liquid\footnote{Note that ${}^{\FZ_1(\CA)}\CA$ contains no further information than $\CA$. That is why we choose to define $\CXsk$ by $\CA$. Also note that ${}^{\FZ_1(\CR)}\CA$ already includes some information of the local quantum symmetry.}, and refer to $\CR$, together with the braided equivalence $\phi: \FZ_1(\CR) \to \FZ_1(\CA)$, as the local quantum symmetry. In other words, we have $\CX=(\CXlqs,\CXsk)$ for any gapped quantum liquid $\CX$.

\begin{rem} \label{rem:sk-diff-QL}
The same topological skeleton can be associated to gapped or gapless quantum liquids depending on what the local quantum symmetries are \cite{KZ20,KZ21,KWZ21}. In \cite{KZ22}, we construct the local quantum symmetries of 2D gapped liquid phases with a finite onsite symmetry as certain topological nets. 
\end{rem}

\begin{figure}
$$
\raisebox{-2.5em}{\begin{tikzpicture}[scale=0.5]
\fill[gray!20] (-3,0) rectangle (3,3) ;
\draw[very thick] (-3,0)--(3,0) node[midway,below] {\footnotesize $\CXsk$} ;
\draw[very thick] (3,3)--(3,0) node[midway,right,scale=0.8] {$\CB$} ;

\draw[fill=white] (2.9,-0.1) rectangle (3.1,0.1) node[midway,below] {\scriptsize $(\CM,m)$} ;
\node[above] at (-0,1.3) {\scriptsize $\FZ_1(\CXsk)$} ;

\end{tikzpicture}}
\quad \xrightarrow{\mbox{\scriptsize topological Wick rotation}} \quad 
\raisebox{-2em}{\begin{tikzpicture}[scale=0.5]
\draw[ultra thick] (-3,1)--(3,1) node[midway,below] {\scriptsize ${}^{\FZ_1(\CXsk)}\CXsk$} ;
\draw[fill=white] (2.9,0.9) rectangle (3.1,1.1) node[midway,above] {\scriptsize $({}^\CB\CM,m)$} ;
\end{tikzpicture}}
$$
\caption{the topological skeleton of the boundary of a quantum liquid}
\label{fig:boundaries}
\end{figure}

If we focus only on the topological skeletons of quantum liquids, our theory can be easily extended to include the topological skeletons of the gapped or gapless boundaries via the representation theory of (enriched) unitary multi-fusion higher categories. We briefly sketch it now. Let $\CX$ be an $n+$1D quantum liquid. Without loss of generality, we can assume that $\CXsk$ is an indecomposable unitary multi-fusion $n$-category. The associated enriched multi-fusion category is ${}^{\FZ_1(\CXsk)}\CXsk$ (see Figure \ref{fig:boundaries}). The topological skeleton of a boundary of $\CX$ is determined by a left $\CXsk$-module $\CM$, together with a distinguished object $m\in\CM$ which specifies the boundary condition. The $\CXsk$-module structure on $\CM$ is defined by a monoidal $*$-functor $\rho: \CXsk \to \Fun(\CM,\CM)$. The enriched higher category associated to the boundary is ${}^\CB\CM$, where $\CB:=\Fun_{\CXsk}(\CM,\CM)$. The enriched higher category ${}^\CB\CM$ is naturally a left ${}^{\FZ_1(\CXsk)}\CXsk$-module \cite{KZ21}. Moreover, the boundary-bulk relation should hold, i.e. $\FZ_0({}^\CB\CM) \simeq {}^{\FZ_1(\CXsk)}\CXsk$ \cite{KWZ15,KWZ17}. This equivalence means that the
left ${}^{\FZ_1(\CXsk)}\CXsk$-module ${}^\CB\CM$ is closed. If $\CX$ is an $n+$1D topological order, then we have $\FZ_1(\CXsk)\simeq n\Hilb$. In this case, if $\rho$ is an equivalence, then the boundary is gapped; otherwise, the boundary is either equipped with some internal symmetry or gapless. 
Since $\CB$ is determined by the $\CXsk$-module $\CM$, we can define the topological skeleton of the boundary simply by a pair $(\CM,m)$. 
This mathematical theory of the topological skeletons of the boundaries of quantum liquids is later proved in a direct study of the Ising chain and the Kitaev chain and their boundaries in \cite{KWZ21} and in more general 2D lattice models in \cite{XZ22}. 

\begin{rem}
The theory of the topological skeletons of defects in higher codimensions can be developed in a similar fashion. This leads us to an explicit construction of the categories of the topological skeletons of quantum liquids. In Subsection \ref{sec:tosk}, however, we compute these categories in a different manner and show that the result is compatible with the results in this subsection. 
\end{rem}

In summary, we have shown that the theory of quantum liquids unifies all topological orders, SPT/SET orders, symmetry-breaking orders and 2D rational CFT's. It is reasonable to believe that quantum liquids also include certain higher dimensional gapless quantum phases whose low energy effective theories are certain higher dimensional CFT's. A quantum liquid can be described by a pair $(\CXlqs,\CXsk)$, where the topological skeleton $\CXsk$ can be viewed as an anomalous topological order and the local quantum symmetry $\CXlqs$ encodes the information of local observables and cancels the anomaly of $\CXsk$. In gapless cases, local quantum symmetries encode the information of certain correlation functions; the topological skeleton encodes all the topological (or categorical) information, such as all topological defects. Together, they can also recover the correlation functions on each gapless defect as in \cite{KZ20,KZ21}. 

We want to emphasize that finding a unified mathematical framework to include both phases within and beyond Landau's paradigm is a necessary step towards a new paradigm. This work allows us to catch a glimpse of the new paradigm. A more complete picture of the new paradigm is developed in \cite{KZ22}. Interestingly, both $\CXlqs$ and $\CXsk$ generalize the notion of a symmetry in the Landau's paradigm (see Remark \ref{rem:LW-model}). 

\begin{rem} \label{rem:LW-model} 
First, it is a long tradition to view $\CXlqs$ in a 2D CFT as a symmetry. Secondly, $\CXsk$ can also be viewed as a higher dimensional symmetry. For example, in a 3D Levin-Wen model constructed from a unitary fusion 1-category $\CC$, the particles in the bulk can be identified with loop operators commuting with the Hamiltonian \cite{LW05}, i.e. 1-dimensional symmetries. Using the fact that a $B_p$-operator realizes the composition of morphisms around a plaquette in $\CC$, it is easy to see that the loop operators commuting with the Hamiltonian one-to-one correspond to $\CC$-$\CC$-bimodule functors, i.e. objects in $\FZ_1(\CC)$ \cite{KK12}. Similarly, particles on a gapped boundary can also be viewed as 1-dimensional symmetries \cite{KK12}. 
 
Levin-Wen models can be generalized to higher dimensions. Given a unitary fusion $n$-category $\CC$, we construct an $n+2$D lattice model in the following way. Fix a triangulation of an $(n+1)$-dimensional space manifold. We label each cell of codimension $k+1$ by a simple $k$-morphism of $\CC$. All such labels span the total Hilbert space. The Hamiltonian contains two type of stabilizers as in the 3D Levin-Wen models. More precisely, the first type of stabilizers project the Hilbert space to a subspace consisting of all composable morphisms, and the second type of stabilizers simply composing all morphisms around a $k$-cell. Then a defect of codimension two can be identified with an $n$-dimensional operator defined on an $n$-sphere and commuting with the Hamiltonian, i.e. an $n$-dimensional symmetry. Such an operator is precisely given by a $\CC$-$\CC$-bimodule functor. Together with all higher codimensional defects, we expect that they form the category $\FZ_1(\CC)$. 

A gapped boundary lattice model can be similarly constructed from an indecomposable unitary left $\CC$-module $\CM$. The defects of codimension one on the boundary can be identified with $n$-sphere operators commuting with the Hamiltonian on the boundary. Together with all the higher codimensional defects, we expect that they form the unitary fusion $n$-category of all $\CC$-module functors from $\CM$ to $\CM$. Therefore, the category of topological defects can be viewed as the category of ``higher dimensional symmetries'', i.e. higher dimensional operators commuting with the Hamiltonian. 

In the literature, a fusion $n$-category is also called an ``algebraic higher symmetry'' (see \cite{JW20,KLWZZ20b}) or ``fusion category symmetry'' (see, for example, \cite{TW19}); and its $E_1$-center is also called a ``categorical symmetry'' in \cite{JW20,KLWZZ20b}. 
\end{rem}

\subsection{The higher categories $\TO^n$} \label{sec:cat-QL}

The notion of a quantum liquid is new. It generalizes the existing physical notion of a gapped quantum liquid \cite{ZW15,SM16} by including certain gapless quantum phases, such as those 2D gapless phases described by 2D rational CFT's. Unfortunately, a microscopic definition of a gapless quantum liquid is not yet available. However, a macroscopic definition is possible because a quantum phase is a macroscopic notion. In principle, it can be defined by all its macroscopic observables. From the unified mathematical theory of gapped and gapless boundaries of 3D topological orders \cite{KZ18b,KZ20,KZ21} and Hypothesis \ref{hyp:quantum-liquid}, we have seen that there is a much bigger mathematical theory of a large family of gapped and gapless quantum phases far beyond topological orders and 2D rational CFT's. We name this family of quantum phases {\it quantum liquids}. We emphasize that, by developing this bigger theory, we automatically obtain a macroscopic definition of a quantum liquid. 

\begin{rem}
To answer which gapless theories, besides 2D rational CFT's, are examples of gapless quantum liquids demands extensive research in the future. However, we believe that 
the critical points of continuous topological phase transitions among SPT/SET orders and symmetry-breaking orders are good candidates for gapless quantum liquids. 
\end{rem}

Our guiding principles are topological Wick rotation and the boundary-bulk relation. In particular, all quantum liquids are required to satisfy the following three conditions that are needed for the proof of boundary-bulk relation \cite{KWZ15,KWZ17}. (1) the anomaly-free bulk of a potentially anomalous quantum liquid is unique; (2) the fusion among quantum liquids (or defects in them) is well-defined (see more discussion in Remark \ref{rem:fusion}); (3) dimensional reductions via fusions are independent of the order of the fusion as the $\otimes$-excision property of factorization homology \cite{AF20}. As a consequence, the boundary-bulk relation holds for all quantum liquids by definition.

\begin{rem} \label{rem:fusion}
Without the microscopic definition of quantum liquids, it is impossible to define the fusion of quantum liquids or defects microscopically. However, in general, instead of squeezing two defects closer to each other\footnote{In general, the fusion of two defects is not just the fusion of local observables living on each defect. When two defects are getting close, more local operators are possible \cite{KZ20,KZ21b}.}, a more controllable way to realize the fusion is by cross-graining. Therefore, the scale invariance is a necessary property of a quantum liquid. By Zamolodchikov-Polchinski Theorem/Conjecture (see for example \cite{Nak14}), under some natural assumptions, scale invariance can be enhanced to conformal invariance. Therefore, higher dimensional CFT's satisfying certain finiteness or fully dualizability are good candidates for quantum liquids. 

On the other hand, the fusion of quantum liquids (or defects) can be defined macroscopically via the tensor product of their topological skeletons as certain categorical structures can be defined mathematically. Examples include the fusion of defects in topological orders \cite{KW14,KWZ15,JF20} and in a fixed 2D rational CFT \cite{FFRS07,DKR15} and among different 2D rational CFT's \cite{KYZ21,KZ21}. The compatibility between the fusion of topological skeleton and that of local quantum symmetries and the compatibility of the microscopic and macroscopic approaches were formulated as a general principle/hypothesis of the universality at RG fixed points in \cite{KZ20}. The second compatibility was known explicitly only in certain gapped lattice models \cite{KK12,KWZ15,BD19,BD20,KTZha20}.
\end{rem}

As we have mentioned in Remark \ref{rem:liquid-like}, our intuition of the `liquid-like' property of quantum liquids and defects is that they are soft enough so that it does not rigidly depend on the local geometry of spacetime similar to 2D rational CFT's and topological defects in 2D rational CFT's \cite{FFRS07}. More precisely, it means that one can bend the phases or defects without making any difference. This already implies certain fully dualizability and produces a fully extended TQFT via cobordism hypothesis \cite{BD95,Lur09}. Moreover, the `liquid-like' property also implies that the quantum liquids in our sense depend covariantly on framing.



\medskip
Now we partially define the higher category $\TO^n$ of $n$D quantum liquids. An object of $\TO^n$ is an $n$D (spacetime dimension) quantum liquid; a 1-morphism is a domain wall; a 2-morphism is a defect of codimension two; so on and so forth. An $n$-morphism is a 0D defect, which is also called an instanton. The composition of morphisms are defined by the fusion of defects (recall Remark \ref{rem:fusion}). Morphisms higher than $n$ are possible and reasonable (see Remark \ref{rem:n+1-morphism-net} for more discussion) but they are not observables in spacetime. At this stage, we do not try to define morphisms higher than $n$ precisely. However, as we proceed, certain precise results of $\TO^n$ naturally emerge.

\begin{rem} \label{rem:n+1-morphism-net}
In a physical realization of a 0D defect $\CV$, the spacetime observables on $\CV$ consists of a spacetime operator $v$ together with the actions of operators $\{ \phi \}$ living in the neighborhood of $\CV$. In other words, it is a pair $(V,v)$, where $V$ is the space of operators invariant under the action of $\{ \phi \}$. Therefore, it is natural to define a morphism $f: (V,v)\to (V',v')$ by a linear map $f:V\to V'$ such that $f$ intertwines the action of $\{ \phi \}$ and $f(v)=v'$. A more precise mathematical definition of $\{ \phi \}$ is available through the theory of defect nets (of operators) near $\CV$ (see \cite[Section 2.2]{KZ22}). In principle, even higher morphisms are possible but they are not directly observable in spacetime. 
\end{rem}

\begin{rem}
In physics literature, a quantum phase almost always refers to an indecomposable (or simple) one. For the study of the categories of quantum liquids, it is important to include the direct sums of them because they naturally appear in the fusions of phases. Similar to composite anyons, a decomposable quantum liquid (or defect) can be called a composite quantum liquid (or defect).  
\end{rem}

The stacking of two quantum liquids defines a symmetric tensor product in $\TO^n$. The trivial $n$D quantum liquid, denoted by $\one^n$, defines the tensor unit. Together, they endow $\TO^n$ with a symmetric monoidal structure. 
We have $\TO^{n-1} = \Omega\TO^n$. Moreover, the time-reversal operator defines an involution $*:\TO^n\to(\TO^n)^{\op n}$. The higher category $\TO^n$ encodes the information of all topological defects in an $n$D quantum liquid $\CA$. For example, we have a monoidal higher category $\Omega(\TO^n,\CA)$, which encodes all the defects of codimension $\ge1$ in $\CA$, and we have a braided monoidal higher category $\Omega^2(\TO^n,\CA)$, which encodes all the defects of codimension $\ge2$ in $\CA$. 



\smallskip

We propose the following hypothesis. 

\begin{hyp} \label{hyp:to-cond}
$\TO^n \simeq \Sigma_*\TO^{n-1}$. 
\end{hyp}

\begin{rem}
If a quantum liquid or a defect can be obtained by $*$-condensation, it is called a $*$-condensation descendant. The category $\TO^n$ should contain all possible $*$-condensation descendants unless there is a physical law forbidding their appearance. Therefore, $\TO^n$ must be $*$-condensation-complete. It remains to show that every $n$D quantum liquid $X\in\TO^n$ is a $*$-condensate of the trivial one $\one^n$ to establish Hypothesis \ref{hyp:to-cond}. Without loss of generality, we can assume that $X$ is indecomposable. Let $f:\one^n\to X$ be a 1-morphism, i.e. a boundary of $X$ as illustrated in Figure \ref{fig:cc-QL}: 
\begin{figure}
\begin{center}
\begin{tikzpicture}[scale=0.8]
\fill[gray!20] (-3,-1.5) node[above right,black] {$X$} rectangle (3,1.5) node[below left,black] {$X$};
\draw[ultra thick, fill=white] (-0.5,0) node[above right] {$\one^n$} (0,0) circle (1);
\node[left] at (1,0) {$a$} ;
\node[right] at (-1,0) {$b$} ;
\node[] at (0,1.25) {$f$} ;
\node[] at (0,-1.2) {$g$} ;
\draw[dashed] (-3,0)   --(-1.1,0) ;
\draw[dashed] (1.1,0) --(3,0) ;
\draw (-2,0) node[above=0pt] {$\Id_X$} ;
\draw (2,0) node[above=0pt] {$\Id_X$} ;
\end{tikzpicture}
\end{center}
\caption{Condensation completion of $\TO^n$}
\label{fig:cc-QL}
\end{figure}
Viewing the boundary as a wall between $X$ and $\one^n$ yields another 1-morphism $g:X\to\one^n$. Moreover, we have evident 2-morphisms $a: f\circ g \to \Id_X$ and $b: \Id_X \to f\circ g$ as shown in the picture, where $\Id_X$ is the trivial wall between $X$ and $X$. The construction goes all the way up to $n$-morphisms and the circle  (actually a cylinder $S^1\times \mathbb{R}^{n-2}$) in Figure \ref{fig:cc-QL} close up to an $(n-1)$-sphere, which defines a 0D defect in spacetime. We believe that a proper-but-yet-unknown definition of higher morphisms allows us to continue this process to extending $(f,g,a,b\cdots)$ to a $*$-condensation $\one^n\condense X$. The same argument can be applied to any two $X,Y\in \TO^n$ and gives a condensation $X\condense Y$. 
\end{rem}


\subsection{The higher categories $\TOsk^n$} \label{sec:tosk}

In this subsection, we focus on the topological skeletons of quantum liquids and defects. 

We denote by $\TOsk^n$ the symmetric monoidal higher category of the topological skeletons of $n$D quantum liquids. That is, an object of $\TOsk^n$ is a potentially anomalous topological order; a $k$-morphism is a potentially anomalous (recall Remark\,\ref{rem:anomalous-wall}) gapped defect of codimension $k$ for $1\leq k\leq n$; and possible higher morphisms. 
The topological Wick rotation is formulated mathematically by the forgetful functor 
$$\TO^n \to \TOsk^n, \quad \CX\mapsto \CX_{\mathrm{sk}}.$$

It turns out that the higher categories $\TOsk^n$ are much more accessible than $\TO^n$. By \cite{KW14,KWZ15,JF20}, a potentially anomalous $n+$1D topological order gives a unitary multi-fusion $n$-category $\CA$ and a potentially anomalous gapped $n$D domain wall gives a unitary bimodule with a distinguished object. If $\FZ_1(\CA)$ is trivial, then $\CA$ is an anomaly-free $n+$1D topological order. In principle, we can write down $\TOsk^n$ explicitly. However, we would like to compute $\TOsk^n$ differently by using Hypothesis\,\ref{thm:TOsk} and show its consistency with above picture. This consistency provides a strong evidence of Hypothesis\,\ref{thm:TOsk}.

Similar to Hypothesis \ref{hyp:to-cond}, we propose the following hypothesis.
\begin{hyp} \label{thm:TOsk}
$\TOsk^n \simeq \Sigma_* \TOsk^{n-1}$. 
\end{hyp}


Note that $\TOsk^0$ can be identified with the coslice 1-category $\C/\Hilb$ described in Example \ref{exam:slice1}, i.e. 
$$
\TOsk^0 \simeq \C/\Hilb.
$$ 
Indeed, a 1D topological order is a unitary multi-fusion 0-category with a trivial center, i.e. an algebra $\End_\C(U)$ where $U\in \Hilb$. A potentially anomalous 0D topological order is precisely a boundary of a 1D topological order, thus can be mathematically described by a pair $(U,u)$, where $u$ is a distinguished element of $U$. Here, the data $u$ is necessary because how the elements of $\End_\C(U)$ are fused into the 0D boundary is a physical data as illustrated in the following picture. 
\begin{center}
\begin{tikzpicture}[scale=0.8]
\draw[fill=white] (-0.1,-0.1) rectangle (0.1,0.1) node[above] {$(U,u)$} ;
\draw[dashed] (-2,0) node[left] {$\C$} -- (-0.1,0) ;
\draw[thick] (0.1,0) -- (3.9,0) ; 
\draw[fill=white] (3.9,-0.1) rectangle (4.1,0.1) node[below=3pt] {$(\Hom_\C(U,V),f)$} ;
\draw[thick] (4.1,0) -- (7,0) ;
\draw (2,0) node[above=0pt] {$\End_\C(U)$} ;
\draw (6,0) node[above=0pt] {$\End_\C(V)$} ;
\end{tikzpicture}
\end{center}
Mathematically, it is natural to define a 1-morphism between two such 0D domain walls $(U,u)$ and $(V,v)$ by a linear map $f: U \to V$ such that $f(u)=v$. This gives us the category $\C/\Hilb$. It turns out that this natural definition of 1-morphisms in $\C/\Hilb$ also has a natural physical meaning. Indeed, a morphism $f: (U,u) \to (V,v)$ between potentially anomalous 0D topological orders can be physically defined by another potentially anomalous 0D topological order $(\Hom_\C(U,V),f)$, together with the following isomorphism 
\begin{align*}
\Hom_\C(U,V) \otimes_{\End_\C(U)} U &\xrightarrow{\simeq} V \\
f\otimes_{\End_\C(U)} u &\mapsto f(u) = v. 
\end{align*}
The mathematical definition and the physical definition are equivalent \cite[Section 5.3]{KWZ15}\cite{KWZ17} (see also \cite[Theorem 3.2.3]{KZ18} for another example of this type of equivalences).  

\begin{rem}
A 0D topological order is an object of the unitary symmetric fusion 0-category $\C$. There is a $*$-equivariant embedding $\C\hookrightarrow\C/\Hilb$, $v\mapsto(\C,v)$.
\end{rem}

\begin{thm} \label{cor:TOts-n}
We have the following explicit mathematical description of $\TOsk^n$: 
$$
\TOsk^n \simeq \left\{ \begin{array}{ll}
\bullet/(n+1)\Hilb \simeq ((n+1)\Hilb/\bullet)^{\op(n+1)}, & \mbox{for even $n$,} \\
(n+1)\Hilb/\bullet \simeq (\bullet/(n+1)\Hilb)^{\op(n+1)}, & \mbox{for odd $n$.}
\end{array} \right.
$$
\end{thm}
\begin{proof}
Combine Proposition \ref{prop:hilb-slice} and Hypothesis \ref{thm:TOsk}. 
\end{proof}


It follows that $n$D potentially anomalous topological orders are classified by pairs $(X,x)$, where $X$ is a unitary $n$-category and $x$ is an object in $X$. By Remark \ref{rem:slice-indecomp}, indecomposable $n$D potentially anomalous topological orders are classified by indecomposable unitary multi-fusion $(n-1)$-categories (i.e. $\Omega(X,x)$). Gapped defects of codimension one are unitary bimodules together with a distinguished object. This is exactly what we expect.  We unravel Theorem\,\ref{cor:TOts-n} in Example \ref{exam:QL1}, \ref{exam:QL2} and \ref{exam:QL3} and show that they are consistent with results in \cite{KZ20,KZ21,KYZ21}.

Moreover, for an indecomposable object $(X,x)\in \TOsk^n$, the physical meaning of $X$ is an anomaly-free $n+$1D topological order with a gapped boundary. When $X$ is viewed as a unitary $n$-category (i.e. an object in $(n+1)\Hilb$),  it is precisely the $n$-category of all gapped boundary conditions of the $n+$1D topological order $X$. The arrow $x:\bullet \to X$ is precisely a gapped boundary of $X$. All defects in $X$ form the fusion $n$-category $\Omega((n+1)\Hilb,X)=\Fun(X,X)$ and those on the boundary $x$ form the multi-fusion $(n-1)$-category $\Omega(X,x)$. Note that $\Omega^2((n+1)\Hilb,X)=\Omega\Fun(X,X)$ is the unitary braided fusion $(n-1)$-category consisting of all topological defects of codimension $\geq 2$ in $X$. 
\begin{cor} \label{cor:bbr}
For indecomposable $(X,x)\in \TOsk^n$, we have the following boundary-bulk relation:
$$
\Omega^2((n+1)\Hilb,X)=\Omega\Fun(X,X) \simeq\FZ_1(\Omega(X,x)).
$$ 
\end{cor}
\begin{proof}
Combine the $*$-variants of Corollary \ref{cor:sigma-hom} and Theorem \ref{thm:mfc-center}. 
\end{proof}

\begin{rem}
The boundary-bulk relation \cite{KWZ15,KWZ17} is our guiding principle when we develop the whole theory. By recovering it from our main result Theorem \ref{cor:TOts-n}, we have passed an important consistency test. Moreover, it also  suggests that the fully dualizability of quantum liquids should suffice Hypothesis\,\ref{hyp:quantum-liquid}. 
\end{rem}


Hypothesis \ref{thm:TOsk} and Theorem \ref{cor:TOts-n} have the following implications. 
\begin{thm}
The symmetric monoidal $(n+1)$-category $\TOsk^n$ is $n$-rigid. In particular, every object of $\TOsk^n$ is $n$-dualizable.
\end{thm}
\begin{proof}
This follows from the $*$-variant of Remark \ref{rem:sigma-rigid}. 
\end{proof}

It follows that, according to the cobordism hypothesis \cite{BD95,Lur09}, every object of $\TOsk^n$ determines an $n$D framed extended TQFT, i.e. a symmetric monoidal functor (see \cite{Lur09} for the precise meaning of the notations)
\begin{equation} \label{eqn:tqft1}
Z: \Bord^\mathrm{fr}_n \to \TOsk^n.
\end{equation}
As we have mentioned in Subsection \ref{sec:cat-QL}, quantum liquids depend covariantly on framing. Therefore, the above functor $Z$ should automatically lift to an oriented extended TQFT
\begin{equation} \label{eqn:tqft2}
Z: \Bord^\mathrm{or}_n \to \TOsk^n.
\end{equation}
This leads to the following mathematical conjecture. See also \cite[Conjecture 1.4.6]{GJF19}.

\begin{conj}
The homotopy $SO(n)$-action on the underlying $(n+1)$-groupoid of $\TOsk^n$ is canonically trivializable.
\end{conj}

By the definition of $\TOsk^n$, all extended TQFT's arising from quantum liquids have the form \eqref{eqn:tqft2}. We expect that quantum liquids catch all the topological information of the spacetime. In mathematical language, this amounts to that the extend TQFT's \eqref{eqn:tqft2} should supply a complete invariant for compact smooth $n$-manifolds as formulated in the following conjecture.

\begin{conj}
If two $n$-morphisms $f$ and $g$ in $\Bord^\mathrm{or}_n$ are not equivalent, then there exists a symmetric monoidal functor \eqref{eqn:tqft2} such that $Z(f)$ and $Z(g)$ are not isometric in $\TOsk^n$.
\end{conj}

\begin{rem}
For a UMTC $\CC$, the delooping $\Sigma_*\CC$ is a unitary fusion 2-category hence defines an object of $\TOsk^3$. The extended TQFT $Z: \Bord^\mathrm{fr}_3 \to \TOsk^3$ associated to this object is essentially the same as the one defined in \cite{Zh17} which is expected to extend the Reshetikhin-Turaev TQFT associated to the UMTC $\CC$ down to dimension zero. In fact, the symmetric monoidal 4-category constructed in \cite{Zh17} is embedded in $\TOsk^3$.
\end{rem}

\begin{rem}
In view of \cite[Theorem 7.15]{JFS17} and Theorem \ref{cor:TOts-n}, the extended TQFT's \eqref{eqn:tqft2} are nothing but the oplax twisted or relative extended TQFT's with target $(n+1)\Hilb$. See also \cite{ST11,FT14,FV15}.
\end{rem}

\subsection{$\TOsk^n$ in low dimensions and 2D CFT's}
Theorem \ref{cor:TOts-n} is heavily loaded. We unravel it for a few low dimensional cases and show that they indeed reproduce the topological skeletons and the categorical information of all macroscopic observables on the 2D worldsheet of all 2D rational CFT's obtained in \cite{KYZ21} and those of gapless boundaries of 3D topological orders obtained in \cite{KZ20,KZ21}. 

\smallskip
According to \cite{KYZ21}, the topological skeletons of 2D rational CFT's form a 3-category 
$\widehat{\mathrm{MFus}^{\mathrm{ind}}_\bullet}$, which is defined and slightly simplified as follows. 
\begin{itemize}
\item A 0-morphism is an indecomposable unitary multi-fusion category $\CA$.
\item A 1-morphism $\CA\to\CB$ is a pair $(\CX,x)$, where $\CX\in \BMod_{\CA|\CB}(2\Hilb)$ and $x\in \CX$ are both nonzero. 
\item A 2-morphism between two 1-morphisms $(\CX,x),(\CX',x') : \CL\to\CM$ is a pair $(F,f)$, where $F:\CX\to\CX'$ is an $\CA$-$\CB$-bimodule $*$-functor and $f:F(x)\to x'$ is a morphism in $\CX'$.
\item A 3-morphism between two 2-morphisms $(F,f),(F',f'):(\CX,x)\to(\CX',x')$ is a bimodule natural transformation $\phi:F\to F'$ such that $f=f'\circ\phi_x$.
\end{itemize}
The physical meanings of the ingredients of $\widehat{\mathrm{MFus}^{\mathrm{ind}}_\bullet}$ is provided by an 3-equivalence\footnote{The same story can be told via the language of enriched categories (see \cite[Section\ 5.4]{KYZ21}).} $\widehat{\FZ}: \widehat{\mathrm{MFus}^{\mathrm{ind}}_\bullet} \to \widehat{\mathrm{BFus_\bullet^{\mathrm{cl}}}}$. The precise definition of the 3-category $\widehat{\mathrm{BFus_\bullet^{\mathrm{cl}}}}$ can be found in \cite[Section 5.2]{KYZ21}. Instead of recalling it, we only explain the image of $\widehat{\FZ}$ which defines the observables on the 2D worldsheet as illustrated in Figure \ref{fig:2-category}.


\begin{figure}
\begin{tikzpicture}[scale=1]
\fill[gray!20] (-5,0) rectangle (5,5) ;
\draw[line width=3.5pt, color=blue] (0,0) -- (0,2.4) node[very near start,left] {\small $[x,x]_\CU$};
\draw[line width=3.5pt, color=blue] (0,2.6) -- (0,5) node[very near end,left] {\small $[x',x']_{\CU'}$};
\draw[color=purple, line width=2pt, fill=white] (-0.1,2.4) rectangle (0.1,2.6) node[midway,left] {\small $(\CW,F,[x,x']_{\CW},\tilde{f})\,\,$} ;
\draw[->,dashed] (-4,2.8) -- (-0.3,4.2) ;
\draw[->,dashed] (-4,2.2) -- (-0.3,0.8) ;
\draw[->,dashed] (4,2.8) -- (0.3,4.2) ;
\draw[->,dashed] (4,2.2) -- (0.3,0.8) ;
\draw[->,dashed] (-0.2,4.2) -- (-0.2,2.7) node[midway,left] {\small $\tilde{f}_l$};
\draw[->,dashed] (0.2,4.2) -- (0.2,2.7) node[midway,right] {\small $\tilde{f}_l$} ;
\draw[->,dashed] (-0.2,0.8) -- (-0.2,2.3) node[midway,left] {\small $\tilde{f}_r$};
\draw[->,dashed] (0.2,0.8) -- (0.2,2.3) node[midway,right] {\small $\tilde{f}_r$} ;
\draw (-4.3,4.2) node[above=0pt] {\small $\FZ_1(\CA)$} ;
\draw (4.3,4.2) node[above=0pt] {\small $\FZ_1(\CB)$} ;
\draw (-5,2.5) node[right=0pt] {\small $[\one_\CA,\one_\CA]_{\FZ_1(\CA)}$} ;
\draw (5,2.5) node[left=0pt] {\small $[\one_\CA,\one_\CA]_{\FZ_1(\CB)}$} ;

\draw (0,0) node[below=0pt] {\scriptsize $(\CX,x)$}; 
\draw (0,5) node[above=0pt] {\scriptsize $(\CX',x')$};

\draw (-4.3,0) node[below=0pt] {\scriptsize $\CA$} ;
\draw (4.3,0) node[below=0pt] {\scriptsize $\CB$} ;

\end{tikzpicture}
\caption{Observables on the 2D worldsheet of rational CFT's}
\label{fig:2-category}
\end{figure}

\begin{enumerate}
\item 0-morphims: we have
\begin{equation} \label{eq:internal-hom}
\widehat{\FZ}: \CA \mapsto (\FZ_1(\CA), [\one_\CA,\one_\CA]_{\FZ_1(\CA)}), 
\end{equation}
where  the internal hom $[\one_\CA,\one_\CA]_{\FZ_1(\CA)}$ is a Lagrangian algebra in $\FZ_1(\CA)$. If a proper local quantum symmetry is provided, this algebra defines a 2D bulk CFT, which contains all the information of the OPE of all bulk fields and that of the modular-invariant correlation functions of all genera. 

\item 1-morphisms: we have 
\begin{equation} \label{eq:two-homs}
\widehat{\FZ}: (\CX,x) \mapsto (\CU:=\Fun_{\CA|\CB}(\CX,\CX), \,\, [x,x]_{\CU}).  
\end{equation}
This image defines a 1D wall CFT (see the blue lines in Figure \ref{fig:2-category}) between two 2D bulk CFT's, and  the internal hom $[x,x]_{\Fun_{\CA|\CB}(\CX,\CX)}$ encode the information of the OPE of fields living in the 1D wall CFT. 

\item 2-morphisms: for $(F,f): (\CX,x) \to (\CX',x')$ and $\CU'=\Fun_{\CA|\CB}(\CX',\CX')$, 
\begin{equation} \label{eq:quadruple}
\widehat{\FZ}: (F,f) \mapsto (\CW:=\Fun_{\CA|\CB}(\CX,\CX'), F, [x,x']_\CW, \tilde{f}),
\end{equation}
where $\tilde{f}: F \to [x,x']_\CW$ is the mate of $f: F(x) \to x'$. Note that this quadruple precisely defines a 0D wall between two 1D wall CFT's (see the purple square in Figure \ref{fig:2-category}). The dashed lines in Figure \ref{fig:2-category} represents the fusion of fields (or observables in the gapped cases) into the 1D walls and the 0D wall. The fusion maps $\tilde{f}_l$ and $\tilde{f}_r$ were defined in \cite[Equation (5.4),(5.5)]{KYZ21} and encode the information of how observables on the 1D wall CFT's fuse into the 0D wall. 

\item The images of the 3-morphisms can also be defined (see \cite[Section 5.2]{KYZ21}). Since they do not show up on the 2D worldsheet, we omit the discussion here. 

\end{enumerate}
Moreover, both 3-categories $\widehat{\mathrm{MFus}^{\mathrm{ind}}_\bullet}$ and $\widehat{\mathrm{BFus_\bullet^{\mathrm{cl}}}}$ are symmetric monoidal. The tensor product in $\widehat{\mathrm{MFus}^{\mathrm{ind}}_\bullet}$ is the usual Deligne tensor product $\boxtimes$ and the tensor unit is $\Hilb$. The 3-equivalence $\widehat{\FZ}$ is symmetric monoidal. 

\smallskip
Now we are ready to explain how various morphisms in $\TOsk^n$ for $n=1,2,3$ encode the topological skeletons and the categorical information of the macroscopic observables on the 2D worldsheet. 
\begin{exam} \label{exam:QL1}
When $n=1$, $\TOsk^1\simeq (\bullet/2\Hilb)^{\op 2}$. The 0-,1-morphisms in $(\bullet/2\Hilb)^{\op 2}$ determines observables in the 1D spacetime of anomaly-free 1D CFT's via the 3-equivalence $\widehat{\FZ}$ restricted to $\Omega\widehat{\mathrm{MFus}^{\mathrm{ind}}_\bullet}$. 
\begin{enumerate}
\item A 0-morphism in $(\bullet/2\Hilb)^{\op 2}$ is a pair $(\CX,x)$, where $\CX\in 2\Hilb$ and $x: \Hilb \to \CX$. 
It is precisely the topological skeleton of 1D anomaly-free CFT (i.e. its 2D bulk CFT is trivial)\footnote{According to \cite[Remark\ 5.2]{KYZ21}, such a 1D anomaly-free CFT often occurs as a consequence of a dimensional reduction process.}. It can be identified with a 1-morphism in $\widehat{\mathrm{MFus}^{\mathrm{ind}}_\bullet}$ if neither $\CX$ nor $x$ is zero. Its image under $\widehat{\FZ}$ (see Eq. (\ref{eq:two-homs}) for $\CA=\CB=\Hilb$) defines the spacetime observables of the 1D anomaly-free CFT associated to the pair $(\CX,x)$. 

\item A 1-morphism $(\CX,x) \to (\CX',x')$ in $(\bullet/2\Hilb)^{\op 2}$ is a pair $(F,f)$, where $F\in \Fun(\CX,\CX')$ and $f: F\circ x\to x'$ is a natural transformation. It is precisely a 2-morphism in $\widehat{\mathrm{MFus}^{\mathrm{ind}}_\bullet}$. Its image is the quadruple in (\ref{eq:quadruple}), which is precisely the observables on a 0D domain wall between two anomaly-free 1D CFT's. 

\end{enumerate}
The compositions of morphisms in $(\bullet/2\Hilb)^{\op 2}$ are compatible with the physical fusions of defects in $\TOsk^1$ due to the functoriality of $\widehat{\FZ}$. As a consequence, we see that the 3-equivalence $\widehat{\FZ}$ is compatible with Theorem \ref{cor:TOts-n} for $n=1$. 
\end{exam}

\begin{exam} \label{exam:QL2}
When $n=2$, $\TOsk^2\simeq \bullet/3\Hilb$. 
\begin{enumerate}
\item A 0-morphism in $\bullet/3\Hilb$ is a pair $(A,a)$, where $A \in 3\Hilb$ and $a:\bullet \to A$ can be viewed as an object in the unitary 2-category $A$. The unitary multi-fusion 1-category $\CA:=\Omega(A,a)$ defines the topological skeleton of the anomalous 2D topological order determined by $(A,a)$. Without loss of generality, we can assume that $A$ is indecomposable. Then $A\simeq \Sigma_*\CA \simeq \RMod_{\CA}(2\Hilb)$. 
Applying topological Wick rotation, we see that $\CA$ (or equivalently, ${}^{\FZ_1(\CA)}\CA$) defines the topological skeleton of an anomaly-free 2D CFT \cite[Section 5.2]{KYZ21} or an anomaly-free 2D gapped liquid phase (recall Remark \ref{rem:sk-diff-QL}). By the 3-equivalence $\widehat{\FZ}$, the internal hom $[\one_{\CA},\one_{\CA}] \in \FZ_1(\CA)$ in the image of $(\CA,\one_{\CA})$ (recall Equation (\ref{eq:internal-hom})) encodes all the information of the OPE in the 2D bulk CFT if a proper local quantum symmetry is provided. 

\item A 1-morphism $(A,a)\to (B,b)$ in $\bullet/3\Hilb$ is a pair $(\CX,x)$, where $\CX\in \Fun(A,B)\simeq \BMod_{\CA|\CB}(2\Hilb)$, $\CB=\Omega(B,b)$ and $x: \CX \circ a \to b$ is a 2-natural transformation (i.e. a right $\CB$-module 1-functor $\CX(a) \to b$). Then $x^\vee: \CB \to \CX$ defines an object in the $\CA$-$\CB$-bimodule category $\CX$. Using $\widehat{\FZ}$ (recall Equation (\ref{eq:two-homs}), the image of $(\CX,x^\vee)$ under $\widehat{\FZ}$ gives precisely a 1D domain wall between two 2D CFT's (or two gapped liquid phases) and the internal hom $[x^\vee,x^\vee]$ encodes all the information of the OPE in the 1D wall CFT if a proper local quantum symmetry is provided. 

\item A 2-morphism $(\CX,x) \to (\CX',x')$ in $\bullet/3\Hilb$ is a pair $(F,f)$, where $F: \CX \to \CX'$ is a 2-natural transformation (i.e. an $\CA$-$\CB$-bimodule 1-functor) and $f: x \to x'\circ F_a$ is a 2-morphism in $B$. This determines a 2-morphism $(F,\tilde{f}^\vee)$ in $\widehat{\mathrm{MFus}^{\mathrm{ind}}_\bullet}$, where $\tilde{f}:x'^\vee \to F_a(x^\vee)$ is canonically defined by $f$ via dualities. Its image of $(F,\tilde{f}^\vee)$ under $\widehat{\FZ}$ gives precisely a 0D domain wall between two 1D domain walls. 

\end{enumerate}
The compositions of morphisms in $\bullet/3\Hilb$ are compatible with the physical fusions of defects. In other words, the 3-equivalence $\widehat{\FZ}$  is compatible with Theorem \ref{cor:TOts-n} for $n=2$. 
\end{exam}

\begin{exam} \label{exam:QL3}
When $n=3$, $\TOsk^3\simeq 4\Hilb/\bullet$. It is similar to the $n=2$ case. So we omit some similar parts, but focus on the appearance of the topological skeletons of the gapless boundaries of 3D topological orders in $\TOsk^3$. 
\begin{enumerate}
\item An object in $4\Hilb/\bullet$ is a pair $(\CX,x)$, where $\CX$ is a unitary 3-category and $x:\CX \to \bullet$ can be viewed as an object in the unitary 3-category $\CX^\op$. The 2-category $\CA_x:=\Hom_{4\Hilb}(x,x)\simeq \Omega(\CX^\op,x)$ is a unitary multi-fusion 2-category, and is precisely the topological skeleton of a 3D quantum liquid associated to $(\CX,x)$. If $\CX$ is indecomposable, we have $\Sigma_* \CA_x^\rev \simeq \CX$ and $\Omega^2(4\Hilb,\CX) \simeq \FZ_1(\CA_x^\rev)$. If $\FZ_1(\CA_x) \simeq 2\Hilb$, then 3D quantum liquid associated to $(\CX,x)$ is an anomaly-free 3D topological order. 

\item A 1-morphism between two pairs $(\CX,x)$ and $(\CY,y)$ is a pair $(f,\phi)$, where $f\in \Fun(\CX,\CY)$ and $\phi: x \to y\circ f$. Without loss of generality, we can assume that $\CX$ and $\CY$ are both indecomposable. We have $\CX\simeq \LMod_{\CA_x}(3\Hilb)$, $\CY\simeq \LMod_{\CA_y}(3\Hilb)$, $\Fun(\CX,\CY)\simeq\BMod_{\CA_y|\CA_x}(3\Hilb)$, $\Fun(\CX,\bullet)\simeq\RMod_{\CA_x}(3\Hilb)$ and $\Fun(\CY,\bullet)\simeq\RMod_{\CA_y}(3\Hilb)$. It is clear that $\phi:x\to y\circ f$ determines an object $\phi$ in the $\CA_y$-$\CA_x$-bimodule 2-category $f$. Then $\Omega(f,\phi)$ defines the topological skeleton of a domain wall between two 3D quantum liquids associated to $(\CX,x)$ and $(\CY,y)$. It is illuminating to consider two special cases. 
\begin{enumerate}
\item When $(\CY,y)=(\bullet,\Id_\bullet)$, $x$ is indecomposable and $\FZ_1(\CA_x) \simeq 2\Hilb$, we have $\Omega\CA_x$ is a UMTC, $f$ is the 2-category of all boundaries of the 3D topological order $(\CX,x)$ and $\phi\in f$ is a single boundary. If $\phi$ is indecomposable, $\CP:=\Omega(f,\phi)$ is a unitary fusion 1-category and precisely the topological skeleton of the boundary $\phi$. The unital $\CA_x$-action on $f$ induces a braided monoidal functor $\beta: \Omega\CA_x^\rev \to\FZ_1(\CP)$. If $\beta$ is an equivalence, then $\CP$ defines the topological skeleton of a gapped boundary of the 3D topological order $(\CX,x)$; otherwise, $\CP$ defines the topological skeleton of a gapless boundary of $(\CX,x)$ and the associated enriched fusion 1-category is ${}^{(\Omega\CA_x^\rev)'_{\FZ_1(\CP)}}\CP$ (recall Subsection \ref{sec:twr}). If $\Omega\CA_x$ is a chiral UMTC, then $\beta$ is never an equivalence. In this case, $(f,\phi)$ is a gapless boundary. This example shows that $\TOsk^3$ contains the topological skeletons (as $\Omega(f,\phi)$) of all gapped and gapless boundaries of all 3D topological orders. 

\item When $(\CY,y)=(\bullet,\Id_\bullet)$, $x$ is indecomposable and $\Omega^2(4\Hilb,\CX) \simeq \FZ_1(\Sigma_*\CE)$ for a symmetric fusion 1-category $\CE$, then the 3D quantum liquid associated to $(\CX,x)$ can be a 3D SPT/SET (or symmetry-breaking) order with the internal symmetry $\CE$, $\CA_x$ defines its topological skeleton, and $f$ is the 2-category of all boundaries of the 3D SPT/SET (or symmetry-breaking) order and $\phi\in f$ is a single boundary, and the associated enriched 2-category is ${}^\CB f$ where $\CB=\Fun_{\CA_x}(f,f)$. We will provide more details of ${}^\CB f$ elsewhere. 
\end{enumerate}

\end{enumerate}
The compositions of morphisms in $4\Hilb/\bullet$ is compatible with the fusion of defects in spacetime. 
\end{exam}


\subsection{Detecting local quantum symmetries} \label{sec:detect-lqs}


As argued in Subsection \ref{sec:twr}, a quantum liquid $\CX\in\TO^n$ can be described by a pair $(\CXlqs,\CXsk)$, where $\CXsk\in\TOsk^n$ is the topological skeleton and $\CXlqs$ is the local quantum symmetry. 
However, it turns out that the local quantum symmetry $\CXlqs$ is not observable in the higher category $\TO^n$ in the sense that different local quantum symmetries may be related by invertible morphisms in $\TO^n$.

For example, consider the non-chiral $E_8$ CFT which we also denote by $E_8$, viewed as a 2D gapless quantum liquid. The boundary $E_8$ CFT supplies a 1-morphism between $E_8$ and the trivial 2D quantum liquid $\one^2$ in $\TO^2$. This 1-morphism turns out to be invertible by considering higher morphisms, which provides a way to gap a narrow strap with two boundaries $E_8$ CFT's (with boundary modes moving in two opposite directions). In other words, we obtain $E_8\simeq\one^2$ in $\TO^2$. Therefore, $E_8$ can not be distinguished from $\one^2$ in $\TO^2$ at all.


In general, there is no hope to recover local quantum symmetries from the equivalence type of the higher category $\TO^n$. Recall that the morphisms of $\TO^n$ are defined by domain walls. The composition of two morphisms are defined by dimensional reduction during which local quantum symmetries are usually considerably reduced. When applying dimensional reduction all the way to dimension zero, the information of local quantum symmetries are completely lost. We propose the following hypothesis:

\begin{hyp} \label{hyp:TOsk=TO}
The forgetful functor $\TO^0\to\TOsk^0$ is an equivalence. 
\end{hyp} 

Invoking Hypothesis \ref{hyp:to-cond} and Hypothesis \ref{thm:TOsk}, we obtain the following physical prediction:
\begin{cor} \label{thm:TOsk=TO}
The forgetful functor $\TO^n\to\TOsk^n$ is an equivalence. 
\end{cor}

It is enlightening to compare $\TO^n$ with the 1-category of Riemannian manifolds and smooth maps. In this 1-category, Riemannian metrics are not observable because two objects are equivalent if and only if they are diffeomorphic. In order to detect Riemannian metrics, one has to modify the 1-category properly. For example, one uses only isometric maps rather than all smooth maps. In the shrunk 1-category, the metric on a Riemannian manifold $M$ can be recovered by the morphisms from the segments of the Euclidean line to $M$, aka geodesic lines. Under the above analog, $\TOsk^n$ is compared with the 1-category of smooth manifolds and smooth maps; local quantum symmetries are compared with Riemannian metrics.

\smallskip

In our situation, to detect local quantum symmetries we have to modify the higher categories $\TO^n$ by separating transparent domain walls from other invertible ones. 
Roughly speaking, a domain wall $\CW$ between two quantum liquids or defects $\CX$ and $\CY$ is transparent if $\CX$ and $\CY$ can be identified by a local unitary transformation of quantum states such that $\CW$ is a trivial domain wall. For example, the 2D chiral $E_8$ CFT is an invertible but not transparent domain wall between an invertible 3D topological order and the trivial 3D quantum liquid.
Once the transparent domain walls are specified, one is able to recover the information of local quantum symmetries in certain categorical structures. We will come back to this issue in a subsequent paper.


\begin{thebibliography}{999}

\bibitem[AF20]{AF20}
D. Ayala, J. Francis, 
{\it A Factorization Homology Primer, Handbook of Homotopy Theory}, 
Chapman and Hall/CRC, New York, NY, U.S.A., 2020, \arXiv{1903.10961}.

\bibitem[BD95]{BD95}
	J. C. Baez and J. Dolan, 
	{\em Higher dimensional algebra and topological quantum field theory}. 
	J. Math. Phys., 36:6073--6105, 1995. 
	\arXiv{q-alg/9503002}.


\bibitem[BCDP22]{BCDP22}
F. Benini, C. Copetti, L. Di Pietro, 
{\it Factorization and Global Symmetries in Holography}, \arXiv{2203.09537}.

\bibitem[BD19]{BD19}
A. Bullivant, C. Delcamp, 
{\it Tube algebras, excitations statistics and compactification in gauge models of topological phases},
Journal of High Energy Physics, 2019 (10), 2019. \arXiv{1905.08673}

\bibitem[BD20]{BD20}
A. Bullivant, C, Delcamp
{\it Excitations in strict 2-group higher gauge models of topological phases}, 
Journal of High Energy Physics, 2020(1), 2020. \arXiv{1909.07937}.

\bibitem[Cha05]{C05}
C.~Chamon,
{\it Quantum glassiness in strongly correlated clean systems: An example of topological overportection}, 
Phys. Rev. Lett. 94 (2005) 040402. 

\bibitem[CGW10a]{CGW10}
X. Chen, Z.-C. Gu and X.-G. Wen,
{\it Local unitary transformation, long-range quantum entanglement, wave function renormalization, and topological orders}, 
Phys. Rev. B 82, 155138 (2010).
\arXiv{1004.3835}.

\bibitem[CGW10b]{CGW10b}
X. Chen, Z.C. Gu, X.G. Wen,
{\it Classification of Gapped Symmetric Phases in 1D Spin Systems}, Phys. Rev. B 83, 035107 (2011)

\bibitem[CGLW13]{CGLW13}
X. Chen, Z.-C. Gu, Z.-X. Liu, X.-G. Wen, 
{\it Symmetry protected topological orders and the group cohomology of their symmetry group}, Phys. Rev. B 87 155114 (2013).
\arXiv{1106.4772}.

\bibitem[CLW11]{CLW11}
X. Chen, Z.-X. Liu, X.-G. Wen, 
{\it Two-dimensional symmetry-protected topological orders and their protected gapless edge excitations}, 
Phys. Rev. B 84 235141 (2011).
\arXiv{1106.4752}.


\bibitem[CJKYZ20]{CJKYZ20}
W.-Q. Chen, C.-M. Jian, L. Kong, Y.-Z. You, H. Zheng,
{\it A topological phase transition on the edge of the 2d $\mathbb{Z}_2$ topological order}, Phys. Rev. B 102, 045139 (2020), \arXiv{1903.12334}



\bibitem[CR16]{CR16}
N. Carqueville and I. Runkel, 
{\it Orbifold completion of defect bicategories}, Quantum Topol. 7 (2016), 203-279 \arXiv{1210.6363}

\bibitem[Dav10]{D10}
A. Davydov, 
{\it Modular invariants for group-theoretic modular data I}, 
J. Algebra 323 (2010) 1321-1348


\bibitem[DKR15]{DKR15}
A. Davydov, L. Kong, and I. Runkel, 
{\it Functoriality of the center of an algebra}, 
Adv. Math. 285 (2015) 811-876, \arXiv{1307.5956}

\bibitem[DM82]{DM82}
	P. Deligne and J. Milne,
	{\em Tannakian categories}. 
	Lecture Notes in Mathematics 900 (1982). 
	\url{http://www.jmilne.org/math/xnotes/tc.pdf}.


\bibitem[DR18]{DR18}
C. L. Douglas and D. J. Reutter, 
{\em Fusion 2-categories and a state-sum invariant for 4-manifolds}, \arXiv{1812.11933}.


\bibitem[DSPS20]{DSPS20}
	C. L. Douglas, C. Schommer-Pries and N. Snyder,
	{\em Dualizable tensor categories}. 
	Memoirs of the AMS, 2020. 
	\arXiv{1312.7188}.


\bibitem[ENO05]{ENO05}
	P. Etingof, D. Nikshych and V. Ostrik,
	{\em On fusion categories}.
	Ann. Math. 162 (2005), 581--642.
	\arXiv{math/0203060}.

\bibitem[FV15]{FV15}
	D. Fiorenza and A. Valentino,
	{\em Boundary conditions for topological quantum field theories, anomalies and projective modular functors}.
	Comm. Math. Phys. 338, 1043-1074 (2015).
	\arXiv{1409.5723}.

\bibitem[FT14]{FT14}
	D. S. Freed and C. Teleman,
	{\em Relative quantum field theory}. 
	Comm. Math. Phys. 326, 459--476 (2014). \arXiv{1212.1692}.

\bibitem[FFRS07]{FFRS07}
J. Fr\"{o}hlich, J. Fuchs, I. Runkel, C. Schweigert,
{\it Duality and defects in rational conformal field theory}, Nucl. Phys. B 763, 354-430 (2007) 

\bibitem[GJF19]{GJF19}
	D. Gaiotto and T. Johnson-Freyd,
	{\em Condensations in higher categories}. 
	2019. 
	\arXiv{1905.09566}.

\bibitem[GW09]{GW09}
Z.-C. Gu, X.-G. Wen, 
{\it Tensor-entanglement-filtering renormalization approach and symmetry-protected topological order}, 
Phys. Rev. B 80, 155131 (2009) \arXiv{0903.1069v2}.

\bibitem[Haa11]{H11}
J.~Haah,
{\it Local stabilizer codes in three dimensions without string logical operators}, 
Phys. Rev. A, 83 (2011) 042330

\bibitem[Hua08]{Hua08}
Y.-Z. Huang, 
{\it Rigidity and modularity of vertex tensor categories}, Commun. Contemp. Math. 10 (2008) 871.

\bibitem[HK07]{HK07}
Y.-Z. Huang, L. Kong, 
{\it Full field algebras}, Commun. Math. Phys. 272 (2007) 345-396, \arXiv{math/0511328}.

\bibitem[JF22]{JF20}
	T. Johnson-Freyd,
	{\em On the classification of topological orders}, Commun. Math. Phys. (2022). https://doi.org/10.1007/s00220-022-04380-3, \arXiv{2003.06663}.

\bibitem[JFS17]{JFS17}
	T. Johnson-Freyd and C. Scheimbauer,
	{\em (Op)lax natural transformations, twisted field theories, and ``even higher'' Morita categories}. 
	Adv. Math., 307:147--223, 2 2017. \arXiv{1502.06526}.

\bibitem[JW20]{JW20}
W. Ji, X.-G. Wen, 
{\it Categorical symmetry and noninvertible anomaly in symmetry-breaking and topological phase transitions}, Phys. Rev. Res. 2, 033417 (2020).

\bibitem[KK12]{KK12}
A. Kitaev, L. Kong,
{\it Models for Gapped Boundaries and Domain Walls}, 
Commun. Math. Phys. 313, 351-373 (2012)



\bibitem[Kon14]{K14}
L. Kong, 
{\it Anyon condensation and tensor categories}, Nucl. Phys. B 886 (2014) 436-482; Erratum and addendum: “Anyon condensation and tensor categories” [Nucl. Phys. B 886 (2014) 436-482], Nucl. Phys. B 973 (2021) 115607; see also a refinement \arXiv{1307.8244v7}.

\bibitem[KLWZZ20a]{KLWZZ20a}
L. Kong, T. Lan, X.-G. Wen, Z.-H. Zhang and H. Zheng,
{\em Classification of topological phases with finite internal symmetries in all dimensions},
J. High Energ. Phys., 2020, 93 (2020)  \arXiv{2003.08898}.


\bibitem[KR09]{KR09}
L. Kong and I. Runkel, 
{\it Cardy Algebras and Sewing Constraints, I}, Commun. Math. Phys. 292, 871-912 (2009) 
\arXiv{0807.3356}

\bibitem[KLWZZ20b]{KLWZZ20b}
L. Kong, T. Lan, X.-G. Wen, Z.-H. Zhang and H. Zheng,
{\em Algebraic higher symmetry and categorical symmetry: A holographic and entanglement view of symmetry}, Phy. Rev. Research, 2, 043086 (2020).
\arXiv{2005.14178}.


\bibitem[KTZho20]{KTZ20}
L. Kong, Y. Tian, S. Zhou, 
{\it The center of monoidal 2-categories in 4D Dijkgraaf-Witten theory}, 
Adv. Math. 360 (2020) 106928  \arXiv{1905.04644v2}


\bibitem[KTZha20]{KTZha20} 
L. Kong, Y. Tian, Z.-H.. Zhang, {\it Defects in the 3-dimensional toric code model form a braided fusion 2-category}, J. High Energ. Phys. 2020, 78 (2020)

\bibitem[KW14]{KW14}
L. Kong and X.-G. Wen, 
{\em Braided fusion categories, gravitational anomalies, and the mathematical framework for topological orders in any dimensions}, \arXiv{1405.5858}.

\bibitem[KWZ15]{KWZ15}
L. Kong, X.-G. Wen and H. Zheng, 
{\em Boundary-bulk relation for topological orders as the functor mapping higher categories to their centers}, \arXiv{1502.01690}.

\bibitem[KWZ17]{KWZ17}
L. Kong, X.-G. Wen and H. Zheng, 
{\em Boundary-bulk relation in topological orders}, Nucl. Phys. B 922 (2017), 62--76.
\arXiv{1702.00673}.

\bibitem[KWZ22]{KWZ21}
L. Kong, X.-G. Wen and H. Zheng, 
{\it One dimensional gapped quantum phases and enriched fusion categories}, 
J. High Energ. Phys. 2022, 22 (2022)
\arXiv{2108.08835}



\bibitem[KYZ21]{KYZ21}
L. Kong, W. Yuan, H. Zheng,
{\it Pointed Drinfeld Center Functor},
Commun. Math. Phys. 381, 1409-1443 (2021), \arXiv{1912.13168}

\bibitem[KYZZ21]{KYZZ21}
L. Kong, W. Yuan, Z.-H. Zhang, H. Zheng,
{\it Enriched monoidal categories I: centers}, \arXiv{2104.03121}


\bibitem[KZ18a]{KZ18c}
L. Kong and H. Zheng,
{\it The center functor is fully faithful}, 
Adv. Math. 339 (2018) 749–779, \arXiv{1507.00503}

\bibitem[KZ18b]{KZ18}
L. Kong and H. Zheng,
{\em Drinfeld center of enriched monoidal categories}, Adv. Math. 323
(2018) 411. \arXiv{1704.01447}.


\bibitem[KZ18c]{KZ18b}
L. Kong and H. Zheng, 
{\it Gapless edges of 2d topological orders and enriched monoidal
categories}, Nucl. Phys. B 927 (2018) 140 \arXiv{1705.01087} 



\bibitem[KZ20]{KZ20}
L. Kong and H. Zheng,
{\em A mathematical theory of gapless edges of 2d topological orders, Part I}, 
J. High Energ. Phys. 2020, 150 (2020). \arXiv{1905.04924}.

\bibitem[KZ21a]{KZ21}
L. Kong and H. Zheng,
{\em A mathematical theory of gapless edges of 2d topological orders, Part II}, 
Nucl. Phys. B 966 (2021), 115384.
\arXiv{1912.01760}.


\bibitem[KZ21b]{KZ21b}
L. Kong and H. Zheng,
{\it Categories of quantum liquids II}, \arXiv{2107.03858}

\bibitem[KZ22]{KZ22}
L. Kong and H. Zheng,
{\it Categories of quantum liquids III}, \arXiv{2201.05726}

\bibitem[LW05]{LW05}
M. A. Levin, X.-G. Wen,
{\it String-net condensation: A physical mechanism for topological phases}, Phys. Rev. B 71, 045110 (2005)

\bibitem[Lur09]{Lur09}
	J. Lurie,
	{\em On the classification of topological field theories}. 
	In Current developments in mathematics, 2008, pages 129--280. Int. Press, Somerville, MA, 2009.
	\arXiv{0905.0465}.

\bibitem[Lur14]{Lur14}
	J. Lurie, 
	{\it Higher Algebra}, 2014, \url{http://www.math.ias.edu/~lurie/papers/HA.pdf}. 


\bibitem[MP17]{MP17}
S. Morrison, D. Penneys, 
{\em Monoidal categories enriched in braided monoidal categories}, 
International Mathematics Research Notices, Vol. 2017, No. 00, (2017) 1--53.
\arXiv{1701.00567}.

\bibitem[Nak14]{Nak14}
Y. Nakayama, 
{\it Scale invariance vs conformal invariance}, Physics Reports, Vol. 569 (2015) 1-93, 
\arXiv{1302.0884}. 


\bibitem[nlab]{nlab}
A list of proposed definitions of a weak $n$-category and references: \url{https://ncatlab.org/nlab/show/n-category}.

\bibitem[ST11]{ST11}
	S. Stolz and P. Teichner, 
	{\em Supersymmetric field theories and generalized cohomology}. 
	In Mathematical foundations of quantum field theory and perturbative string theory, volume 83 of Proc. Sympos. Pure Math., pages 279--340. Amer. Math. Soc., Providence, RI, 2011. 
	\arXiv{1108.0189}.

\bibitem[SM16]{SM16}
B. Swingle and J. McGreevy, 
{\it Renormalization group constructions of topological quantum liquids and beyond}, 
Phys. Rev. B 93 (2016) 045127 \arXiv{1407.8203}

\bibitem[SPGC11]{SPGC11}
 N. Schuch, D. Pe\'{r}ez-Garc\'{i}a, I. Cirac, 
 {\it Classifying quantum phases using matrix product states and projected entangled pair states}, 
 Phys. Rev. B 84 (2011) 165139

\bibitem[TW19]{TW19}
R. Thorngren and Y. Wang, 
{\it Fusion Category Symmetry I: Anomaly In-Flow and Gapped Phases}, \arXiv{1912.02817}


\bibitem[XZ22]{XZ22}
R. Xu and Z.-H. Zhang, 
{\it Categorical descriptions of 1-dimensional gapped phases with abelian onsite symmetries}, \arXiv{2205.09656} 


\bibitem[Wen90]{Wen90}
X.-G. Wen, 
{\it Topological orders in rigid states}. Int. J. Mod. Phys. B 4 (1990) 239--271. 

\bibitem[Wen02]{Wen02}
X.-G. Wen, 
{\it Quantum orders and symmetric spin liquids}, Phys. Rev., B65, 165113 (2002).
\arXiv{cond-mat/0107071}.

\bibitem[Wen17]{Wen17}
X.-G. Wen, 
{\it Zoo of quantum-topological phases of matter}, Rev. Mod. Phys. 89, 41004 (2017).
\arXiv{1610.03911}.

\bibitem[Wen19]{Wen19}
X.-G. Wen, 
{\em Choreographed entanglement dances: Topological states of quantum matter}, 
Science 363 (2019) eaal3099.
\arXiv{1906.05983}.

\bibitem[ZW15]{ZW15}
B. Zeng, X.-G. Wen, 
{\it Gapped quantum liquids and topological order, stochastic local transfor- mations and emergence of unitarity}, Phys. Rev. B 91 (2015) 125121


\bibitem[Zhe17]{Zh17}
	H. Zheng,
	{\em Extended TQFT arising from enriched multi-fusion categories}.
	\arXiv{1704.05956}.

\end{thebibliography}
\end{document}